\documentclass[lettersize,journal]{IEEEtran}
\IEEEoverridecommandlockouts
\usepackage{cite}
\usepackage{amsmath,amssymb,amsfonts}

\usepackage{mathtools}

\usepackage{cuted}
\usepackage{booktabs}
\usepackage{siunitx}

\usepackage{graphicx, caption, subcaption}
\usepackage{textcomp}
\usepackage{xcolor}

\usepackage{amsmath, amssymb, amsthm,algorithm}

\newcommand{\minbox}[2]{%
  \mathmakebox[\ifdim#1<\width\width\else#1\fi]{#2}}

\usepackage[latin1]{inputenc}

\usepackage{latexsym}
\usepackage{graphics,epsfig}
\usepackage{amsfonts}
\usepackage{booktabs}
\usepackage{color}
\usepackage{multirow}
\usepackage[justification=centering]{caption}
\usepackage{dsfont}
\usepackage{adjustbox}
\usepackage{kantlipsum}
\usepackage{cases}
\usepackage{url}
\usepackage{tabu}
\usepackage{makecell}
\usepackage{soul}

\usepackage{array}
\usepackage{algpseudocode}
\usepackage{mathtools,lipsum,cuted}

\usepackage[justification=centering]{caption}

\usepackage[T1]{fontenc}
\usepackage{mwe}

\newcommand\Tstrut{\rule{0pt}{2.6ex}}         
\newcommand\Bstrut{\rule[-0.9ex]{0pt}{0pt}}   

\newlength{\tempheight}
\newlength{\tempwidth}

\newcommand{\rowname}[1]
{\rotatebox{90}{\makebox[\tempheight][c]{#1}}}

\newcommand{\columnname}[1]
{\makebox[\tempwidth][c]{#1}}

\makeatletter
\def\BState{\State\hskip-\ALG@thistlm}
\makeatother

\newcounter{example}[section]

\usepackage[english]{babel}
\usepackage{verbatim}

\usepackage[english]{babel}

\theoremstyle{plain}
 \newtheorem{thm}{Theorem}

 \newtheorem{prop}{Proposition}

\theoremstyle{remark}

\def\BibTeX{{\rm B\kern-.05em{\sc i\kern-.025em b}\kern-.08em
    T\kern-.1667em\lower.7ex\hbox{E}\kern-.125emX}}
\begin{document}

\voffset=0.05in
\textheight=9.28in

\title{Signaling Rate and Performance of \\RIS Reconfiguration and Handover Management \\ in Next Generation Mobile Networks}

\author{\IEEEauthorblockN{Mounir Bensalem and Admela Jukan}\\
\IEEEauthorblockA{Technische Universit\"at Braunschweig, Germany;
\{mounir.bensalem,  a.jukan\}@tu-bs.de}
}

\maketitle

\begin{abstract}
We consider the problem of signaling rate and performance for an efficient control and management of RIS reconfigurations and handover in next generation mobile networks. To this end, we first analytically determine the rates of RIS reconfigurations and handover using a stochastic geometry network model. We derive closed-form expressions of these rates while taking into account static obstacles (both known and unknown), self-blockage, RIS location density, and variations in the angle and direction of user mobility. Based on the rates derived, we  analyze the signaling rates of a sample novel signaling protocol, which we propose as an extension of an handover signaling protocol standard in mobile networks.  The results quantify the impact of known and unknown obstacles on the RIS and handover reconfiguration rate as function of device density and mobility.  We use the proposed analysis to evaluate the signaling overhead due to RIS reconfigurations, as well as to dimension the related RIS control plane server capacity in the network management system.  To the best of our knowledge, this is the first analytical model to derive the closed form expressions of RIS reconfiguration rates, along with handover rates, and relate its statistical properties to the signaling rate and performance in next generation mobile networks.
\end{abstract}


\section{Introduction}

The demand for high capacity of wireless connections has become increasingly critical in the evolution of 5G/6G mobile networks. To this end, mm-wave and terahertz (THz) frequency bands have emerged as pivotal communication technologies to guarantee Quality of Service (QoS) in 6G networks \cite{tripathi2021millimeter}. Both are characterized by high capacity of communication links, high directionality of the antennas, and also high sensitivity to signal attenuation due to obstacles.  To address the issue of blockages,  reconfigurable intelligent surfaces (RISs) and relays are employed to assist the communication by creating alternative paths. By employing passive elements that can manipulate electromagnetic waves, RISs offer opportunities to mitigate channel impairments and improve signal propagation. The incorporation of RISs holds the potential to transform the functioning of wireless networks, providing numerous advantages that could lead to an enriched communication experience for users, such as  digital twin, holographic telepresence, and immersive applications. 

From the signaling and network management perspective the embedding of the RISs is an open issue. Most previous work focused on the wireless channel properties in presence of RISs and assumed that reconfigurations of wireless paths over different RISs were seamless and passive, i.e., did not require signaling effort. In reality, however, when a wireless path needs to be re-computed and re-established, it requires activation of a signaling protocol from the corresponding network management functions. With help of signaling, an active path with RIS needs to be switched to an alternative path with another RIS, even in cases where both connections are established over the same base station. The signaling rate for the activation of new paths depends on the RIS reconfiguration rate, i.e., number of times a new RIS needs to be allocated, which needs to be either measured, estimated or predicted. In case multiple base stations are involved, a handover maybe necessary.  Both RIS reconfigurations and handover necessitate the dimensioning of the corresponding servers in the control and management, such as for the Mobility Management Entity (MME) functions in the current standards. In fact, the related entities in LTE and NR mobility architecture, e.g., 3GPP TS 23.501 \cite{etsi2018LTE} and ETSI TS 136.420 \cite{etsi20185g} standards need to be extended to control and manage the allocations of RISs.

In this paper, we provide a novel analysis of RIS reconfiguration and handover rates, whereby RIS reconfigurations and handovers are considered in course of the dynamic reallocation of inadequate THz/mmW wireless paths in response to the availability of alternative paths with sufficient quality. Based on a stochastic geometry model, we analyze these rates under various conditions and derive closed-form expressions that include static obstacles, consideration of self-blockage, RISs location density and variations in angle and speed distributions in a generic wireless mobile network, including multiple RISs and base stations. We finally propose and analyze the related signaling rates of a sample signaling protocol, designed as an extension of today's known handover signaling standard.  The RIS reconfiguration probability analysis is validated by simulations for two mobility models (random waypoint and random direction mobility models). The results quantify the impact of obstacles on the RIS and handover reconfiguration rate as function of device density and mobility. The proposed analysis are shown as useful when dimensioning and evaluating the signaling overhead due to RIS reconfigurations, which directly impacts the planning of the related control plane server capacity in the network management system. 

To the best of our knowledge, this is the first analytical model to derive closed form expressions of RIS reconfiguration rates along with handover rates, and relate its statistical properties to the signaling rate and performance in next generation mobile networks. In a nutshell, the paper provides the following novel contributions:
 \begin{itemize}
\item Closed form analysis of RIS reconfiguration rate and handover rate due to mobility and obstacles. 
\item Analysis of the inter dependencies of RIS reconfiguration -and handover rates in the mobile network.
\item Evaluation of the impact of mobility on the receiver node in terms of channel quality, in presence of both known and unknown obstacles
\item Design of a sample signaling protocol, as an extension of the current mobile network standard, to actuate handover and RIS reconfigurations, and the related analysis of the signaling performance.
\end{itemize}

The rest of the paper is organized as follows. Section \ref{sec:relwork} presents related work. Section \ref{sec:arch} introduces the reference network architecture, including data and signaling planes. Section \ref{sec:sysmodel} derives analytically the RIS reconfiguration and handover rates as well as the it presents the signaling protocol and analysis its performance. Section \ref{sec:results} presents the theoretical and simulation results. Section \ref{sec:conclusion} concludes the paper.

\begin{table*}[t]
\centering
\begin{tabular}{SSSSSSSSSS} \toprule
    {Reference} & { Single eNB} & {Multiple eNB} & {Single RIS} & {Multiple RIS}& {HO} & {RR} & {SO} & {SB}& {Sig} \\   \midrule
    { \cite{zhang2024adaptive}} & {\checkmark }& & {} & {\checkmark} &  & {\checkmark}  &  {\checkmark}   &  &  \\
    {  \cite{palitharathna2024handover}} &  {}  & {\checkmark} & {\checkmark} & &  {\checkmark} &  {\checkmark}  & {\checkmark}  &  & \\
    {  \cite{zhang2021reconfigurable}} &  {}  & {\checkmark} & {\checkmark} & {} & {}  &  {\checkmark}  & {}  & {}  & {} \\
    {  \cite{you2022deploy}} &  {\checkmark}  & {} & {} & {\checkmark} & {}  &  {\checkmark}  & {}  & {}  & {} \\
 
   { \cite{liu2024sustainable}} &  {\checkmark}  & &  & {\checkmark} &   &  {\checkmark}  &   &  & \\
    
 {  \cite{wei2023equivalent, zhang2024discrete}}&   & {\checkmark}&    & {\checkmark} &  {\checkmark}  &  {\checkmark}  &  &   & \\
 {  \cite{nor2024mobile}}&   &{\checkmark} &  & {\checkmark}  &  &  {\checkmark}  &  {\checkmark}  &   {\checkmark}  &  \\
  {  \cite{10462513}}&   &{\checkmark} &  & {\checkmark}  &  &  {\checkmark}  &  {\checkmark}  &     &  \\
     { \cite{joshi2019association, okaf2020analysis}}  &  & {\checkmark}  &  &  & {\checkmark}  & &   {\checkmark}  & &  \\
     { \cite{jiao2021enabling}}  &   & {\checkmark} &  {\checkmark}  & &  {\checkmark}  &  &  {\checkmark}   & & \\ 
     {\cite{okaf2021analysis, iqbal2023analysis}}  &    &{\checkmark} & & & {\checkmark}  &  &   &   {\checkmark} & \\   
     {  \cite{ulvan2010study,NR2019}}&   & {\checkmark} & & &  {\checkmark}  &  &       & &  {\checkmark} \\
     {\cite{mollel2019handover, ghosh2020analyzing }}  &   & {\checkmark} & & &  {\checkmark}  &  &    {\checkmark}  & &  {\checkmark}  \\   \midrule 
    {Our paper}  &  {} &  {\checkmark} &  {} & {\checkmark} & {\checkmark}  &  {\checkmark}  & {\checkmark}  &  {\checkmark}  & {\checkmark}   \\ \bottomrule
    
\end{tabular}
\caption{Related work and our contribution. Abbreviations: evolved NodeBs (eNB), Handover (HO), RISs, RIS Reconfiguration (RR), Static Obstacle (SO), Self-Blockage (SB), and Signaling (Sig) }
\label{tab:review}
\end{table*}

 \section{Related Work} \label{sec:relwork}
Table \ref{tab:review} summaries the main contributions of this paper with respect to the current literature. We list the work based on considerations of base stations (BS) or evolved NodeBs (eNB), Handover (HO), RISs, RIS Reconfiguration (RR), Static Obstacle (SO), Self-Blockage (SB),  and Signaling (Sig).  

 A generic wireless mobile network including   multiple base stations (BSs) and multiple RISs, was studied in  \cite{ozdogan2019intelligent, zhang2021reconfigurable, you2022deploy}. Similar to  \cite{ozdogan2019intelligent}, we make assumptions of geometic distributions, but also consider the mobility of users. We also show that paths always need to be reconfigured when RIS is placed closer to the receiver.  Majorly related to our work is a stochastic geometry-based model proposed in  \cite{zhang2021reconfigurable} which analyses important metrics of the ergodic capacity and the coverage probability. We use this work  as foundation to analyze the path loss, capacity of links under different conditions, but extend \cite{zhang2021reconfigurable} to consider the mobility as well as to derive the rate of RIS reconfigurations. We also make assumptions proposed in  \cite{you2022deploy}, where the path loss parameters are considered dependent on the deployment of RISs with respect to the BS, which impacts the path loss performance. 

In regard to RIS reconfigurations, for both active and passive RISs, we need to change the phase shift during reconfigurations, whereby the active RIS is able to add power to enhance the signal (also referred to as relay). The removal of blind spots and the ability to activate alternative links to users plays a major role in case of a bad channel quality or blockages \cite{zhang2024adaptive}. Paper \cite{zhang2024adaptive} proposes a learning approach to adaptively switch among RISs  (i.e., RIS reconfiguration) in real time, considering the existence of  obstacles, albeit only for an area with a single base station (BS). In our work, we consider multiple BS for both the HO and RR rates. The RIS reconfiguration is also affected by the accessibility of a path between the eNB-RIS-UE, where the placement of the RISs plays a major role \cite{ozdogan2019intelligent}. Some placements however increase the blockages probability, due to  obstacles, and self-blockages. Papers \cite{bai2014coverage, andrews2016modeling} investigated the impact of  correlation in the blockage event of coexisting links, however without RISs or relays. These papers motivate our analytical model where we also quantified the impact of obstacles on the communication performance, but we included RISs in the analysis.

In regard to handover, paper \cite{jiao2021enabling} proposed  RIS-aided handover scheme in mmW networks using deep reinforcement learning (DRL) to reducing frequent handovers caused by link blockages. In \cite{palitharathna2024handover}  RISs were used for handover management in  an indoor visible light communication  system, whereby the goal is to route the signals through RIS device  to redirect the communication links when the direct links are blocked by obstacles. Recently, an important research effort provided analysis of the impact of blockages on HO rate  in 5G cellular networks, where  the mm-wave Base Stations (BSs) are chosen based on signal quality.  Papers \cite{okaf2020analysis, jain2019impact} majorly focused on the theoretical analysis that provides an expression of the HO, whereas papers \cite{joshi2019association, okaf2021analysis, iqbal2023analysis} evaluate the impact of blockages on handover.  In \cite{okaf2020analysis}, the impact on HO rate was analyzed for blockages due to one obstacle, whereby  the location of BSs was generated using  the homogeneous Poisson Point Process, while the mobile user is moving with a certain speed and direction angle. While this paper assumed that an handover from a BS to another BS is possible to alleviate performance issues, in our work we consider RISs, and the related assumptions about blockages, and self blockage, to analyze both the RIS reconfiguration rate and handover rate.

Signaling procedure in a scenario that considers the existence of RISs has not been studied in any previous work and has not been defined in any standard focused on implementing RIS hardware and controllers in wireless network \cite{liu2024sustainable}, to the best of our knowledge. Paper \cite{liu2024sustainable} presented the state-of-the-art of the work status led  by the ETSI ISG RIS, presenting the  typical deployment scenarios of RISs.   Table \ref{tab:standard} gives an overview  of the published and ongoing ETSI technical reports that focus on RIS integration into the network. They focus mostly on physical layer challenges and requirements for RIS deployment and providing a general view on the communication models, channel models, and channel estimation, i.e., do not include signaling and network management aspects. The reference standards on signaling solutions are 3GPP LTE \cite{ulvan2010study} and 5G New Radio (NR) architecture\cite{NR2019}. In the literature,  \cite{mollel2019handover, ghosh2020analyzing } studied the HO performance using signaling procedure proposed in the standards. 

In our work, we extend the network scenario to RISs, and propose an updated mobility signaling protocol to this end, of which the performance we also analyze. To the best of our knowledge, no analytical model so far was used to derive closed form expressions of RIS reconfiguration rates along with handover rates, to relate their statistical properties to the signaling rate and performance. More in detail, we extend the handover (HO) signaling protocol in 3GPP LTE \cite{ulvan2010study} and 5G New Radio (NR) architecture\cite{NR2019}  to account for RIS reconfigurations and HO. We introduce the key features and entities regarding mobility and  RIS reconfiguration. We propose the concept of serving and neighboring RIS, akin to the concepts of serving and neighboring base stations. The signaling protocol is activated at the derived rate from the analysis and includes both handover and RIS reconfiguration functions. We use the network management functions in LTE and NR mobility architecture as well as the related interfaces, based on 3GPP TS 23.501 \cite{etsi2018LTE} and ETSI TS 136.420 \cite{etsi20185g} standard to examplify a network management system which includes a novel server, related to as RIS-M, akin to the Mobility Management Entity (MME). We not only analytically derive the resulting signaling rate, but we also use it to propose to dimension the related server size in the prospective network management architecture.

\begin{table}
\centering
\begin{tabular}{SSS} \toprule
    {ETSI Document} & {Scope} & {Date of publication} \\ \midrule
  { GR RIS 001\cite{ETSIGRRIS001}} & {Deployment scenarios} & {Apr 2024} \\ 
  & {and requirements}& \\ \midrule
   {GR RIS 002\cite{ETSIGRRIS002}}  & {Deployment challenges} & {Aug 2023}   \\
    &{Impact on network}& \\
    &{ architecture}& \\ \midrule
 {GR RIS 003\cite{ETSIGRRIS003}}  & {Communication model} & {Jun 2023
 }  \\
 &{Channel models}&\\ 
 &{Channel estimation} &\\ 
 & {Evaluation methods} & \\\midrule
 {GR RIS 004\cite{ETSIGRRIS004}}& {RIS implementation} & {Aug 2024}  \\ 
 & {and practical considerations} & \\\midrule
    {GR RIS 005\cite{ETSIGRRIS005}}  & {Analysis of Diversity} & {Jul 2024}  \\
    &{ and Multiplexing} & \\ \midrule
    {GR RIS 006\cite{ETSIGRRIS006}}  & {Solutions for Multi-} & {Sep 2024}   \\
    &{Functional RIS (MF-RIS)}&\\ 
    &{Channel modelling}&\\ \midrule
    {GR RIS 007\cite{ETSIGRRIS007}}   &  {Near-field channel} & {Dec 2025}   \\
    &{modeling}&\\ \bottomrule
    
\end{tabular}
\caption{Published and ongoing ETSI documents studying the implementation of RISs in the network.}
\label{tab:standard}
\end{table}

%
\section{Reference Architecture} \label{sec:arch}
\subsection{RIS Reconfiguration vs. Handover Management}
In this paper, we distinguish between RIS reconfigurations and handover management based on the following. \emph{RIS reconfiguration} is defined as the dynamic reallocation of an active path with RIS to an alternative path with another RIS, whereby both connections are established over the same base station. The \emph{handover,} on the other hand, involves the connections over different base stations. Figure  \ref{fig:RIS_HO_RRR} illustrates the RIS reconfigurations  and handover as an extension to the standard HO process from 3GPP LTE \cite{ulvan2010study} and 5G New Radio (NR) architecture\cite{NR2019}. Here, due to mobility, UE can stay connected either via RIS reconfiguration or handover (HO). In case of RIS reconfigurations, the wireless path is reconfigured from the serving RIS to the target RIS; for HO, the same is implemented from the serving eNB to the target eNB. In  Figure \ref{fig:RIS_HO_RRR}, we assume that both the source and target cells are located in the same eNB (Intra-eNB HO), whereby UE is assigned to a direct link to the serving eNB, while otherwise a path through the target RIS is allocated. 

\begin{figure}[H]
 \centering 
   \includegraphics[scale=0.53]{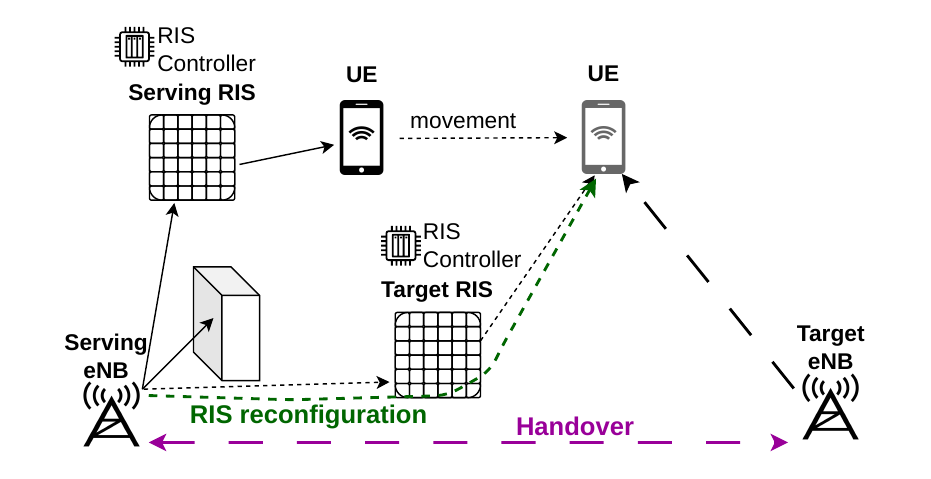}
 \caption{RIS reconfiguration vs. handover (HO) management. }
\label{fig:RIS_HO_RRR}
\end{figure}


\subsection{Modeling statistical properties}
Both RIS reconfigurations and HO are triggered when the wireless path reallocation is necessary due to the inadequate quality and in response to the availability of alternative paths with sufficient quality. The rate of these triggering events is what we refer to as RIS reconfiguration and handover rates. To model these rates statistically, the BS positions in an area are modeled as a homogeneous Poisson Point Process (PPP) with a density $\lambda_{BS}$, akin to \cite{jain2019impact, okaf2021analysis}. We consider that around each BS/eNB there is a set of RISs that can be used to create alternative paths for high frequency signals possibly blocked by obstacles. We assume that RISs are also located in the area around the BS and follow the PPP distribution with a RIS density $\lambda_{RIS}$.  In order to consider  obstacles like walls, buildings, billboard, etc, we assume that the  obstacles are distributed in the area following the PPP distribution with a density $\lambda_B$, and it is assumed to be known \emph{a priori}. Moreover, we assume that an UE can also be self-blocked, e.g., the human body or a physical shape of an moving object can create an obscure angle that blocks signal. The detailed modeling and the impact of those properties on the HO, RIS reconfigurations and signaling rates are presented in Section \ref{sec:sysmodel}.

\subsection{Management and signaling}

To analyze the RIS reconfiguration and HO signaling we based our architecture on LTE and NR mobility  architecture as shown in Figure \ref{fig:RIS_HO_RRR_arch}. Also the related interfaces are based on 3GPP TS 23.501 \cite{etsi2018LTE} and ETSI TS 136.420 \cite{etsi20185g} standards. In this architecture, \textbf{eNB} controls part of the user mobility, in a CONNECTD mode, where it decides if an HO is needed, and then it signals the target eNB for admission control. We propose a new entity called \textbf{RIS Manager (RIS-M)}, which controls the RIS availability and decides the UE admission to a path including a specific RIS.  \textbf{RIS Controller}, on the other hand, controls the phase shift and configurations of RIS devices (whether active or passive), and provides the interface for hardware configurations.  \textbf{Mobility Management Entity (MME)} deals with  the mobility  of users during session setup, IDLE\_MODE, and generally supports subscriber authentication, tracking area management, handovers and roaming.  \textbf{Serving Gateway (SGW)}  is responsible on setting up a user plane and acts as a local anchor for CONNECTED\_MODE user mobility.  SGW sends the user data to UE through eNB.  \textbf{Packet Data Network Gateway (PGW)} connects the LTE to external networks and allocates the IP address to UE.  
 \textbf{ S1}  is an interface between RAN and the LTE core network side. The S1-U interface carries user data between eNB and SGW, and S1-MME interface carries control signaling messages between eNB and MME. S1 can carry HO signaling and HO preparation (HOP) messages between two eNBs, as shown in Fig. \ref{fig:RIS_HO_RRR_arch}. In this case, an HO is called S1-based HO. Finally, \textbf{X2} is an  interface between two eNBs to exchange  control information messages. When X2 interface is used for HO, the HO signaling and HOP messages are transmitted using X2 interface between the source and target eNBs.

Figure \ref{fig:RIS_HO_RRR_sig} illustrates the related L3 signaling for RIS reconfiguration and HO. First,  the serving eNB transmits DL RS to UE through both the serving and neighboring RIS. Afterwards,  UE performs and processes signal strength   measurements, upon which it sends the measurement report (MR) from both paths to the serving eNB. Based on the received measurement report, eNB  makes the decision about RIS reconfigurations and HO, and sends either a RIS reconfiguration request to the RIS-M who makes the RIS admission decision,  or an HO request to the target eNB for an HO admission.\\ 

\begin{figure}[H]
 \centering 
   \includegraphics[scale=0.53]{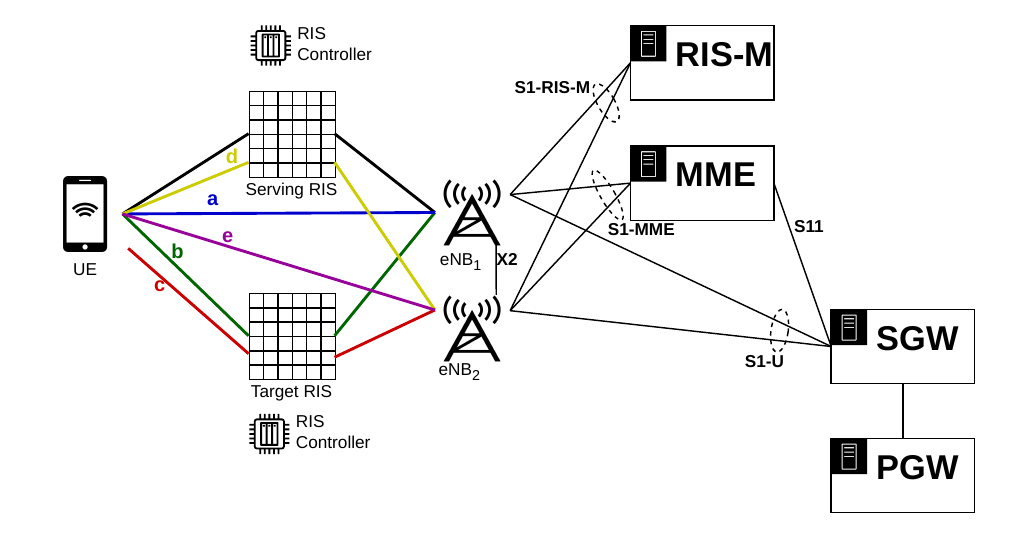}
 \caption{Enhancing the LTE mobility architecture with RIS management entity. MME: Mobility Management Entity, SGW: Serving Gateway, PGW: Packet Data Network Gateway, RIS-M: RIS Manager}
\label{fig:RIS_HO_RRR_arch}
\vspace{-0.5 cm}
\end{figure}

 \begin{figure}[H]
 \centering 
   \includegraphics[scale=0.53]{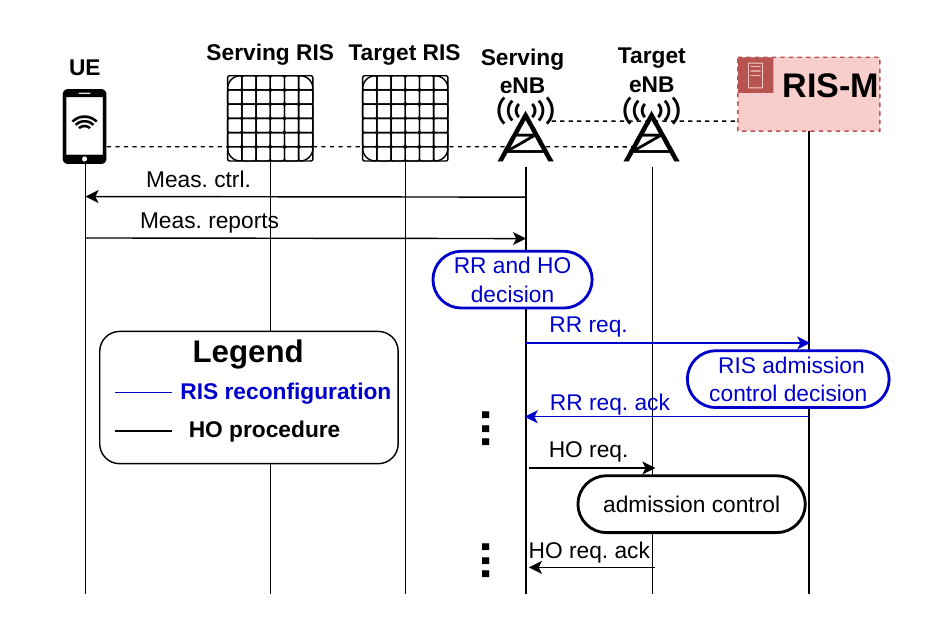}
 \caption{Conceptual proposal for signaling. }
\label{fig:RIS_HO_RRR_sig}
\end{figure}

\section{Analysis}\label{sec:sysmodel}

The goal of this section is to obtain the signaling rate as function of RIS reconfiguration- and HO rates. To this end,  we develop a statistical model of signal quality in a generic THz/mmW mobile network with consideration of mobility and two types of obstacles: static (e.g., walls, buildings, machines) and mobile (e.g., human users, cars, robots). The signaling process in the management plane is triggered when a connection needs to be reconfigured, either due to the user getting closer to a new RIS element than to its operating RIS (RIS reconfiguration), or closer to another eNB (handover).  The notations used in this section are presented in Table \ref{tab:notations}.


\begin{figure}
 \centering 
   \includegraphics[scale=0.4]{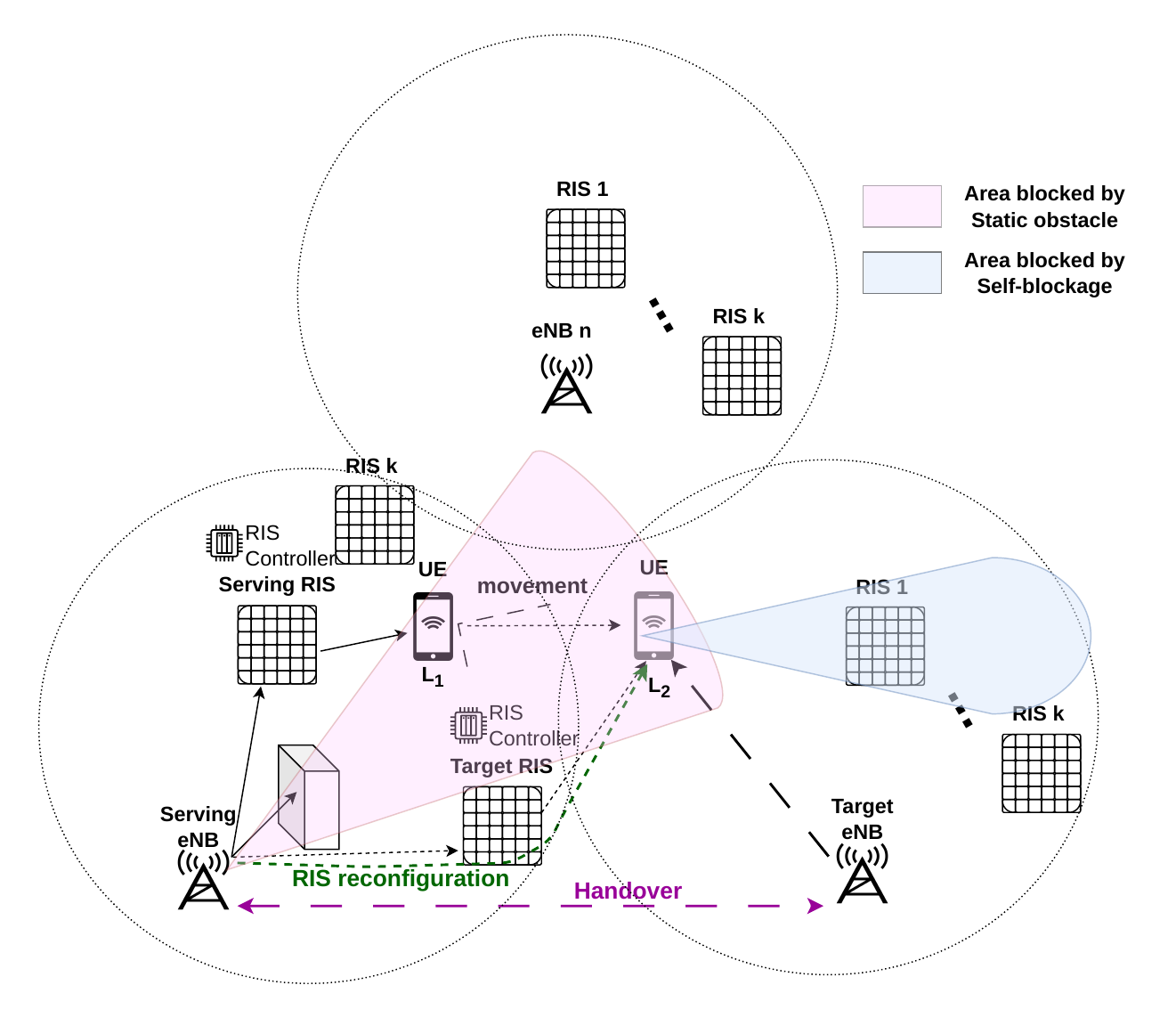}
 \caption{Reference scenario. }
\label{fig:scenario_RIS_eNB}
\vspace{-0.5 cm}
\end{figure}

\begin{table}[h!]
  \begin{center}
    \caption{Notations}
    \label{tab:notations}
    \begin{tabular}{|l|l|}
    \hline 
      \textbf{Notation} & \textbf{Description} \Tstrut\Bstrut \\
      \hline
      $n$ &   Number of eNB \Tstrut \\
      $k$ &   Number of RISs assigned to an eNB  \\
      $N_{SGW}$ & Number of SGW servers\\
      $N_{RISM}$ & Number of RIS-M servers\\
      $R_{LoS}$  & Radius of LoS \\
      $r_{RIS}$ & UE-serving RIS distance before movement \\
      $R_{RIS}$ &  UE-serving RIS distance after movement\\
      $r_{eNB}$ &  UE-serving eNB distance before movement \\
      $R_{eNB}$ & UE-serving eNB distance after movement\\
      $\lambda_{RIS}$ & RISs density    \\ 
      $\lambda_{eNB}$ &  eNBs density  \\ 
      $\lambda_{B}$ & Static obstacle density    \\ 
        d$_U$ & UE movement speed \\ 
        $\lambda_k$ & The mean session arrival  rate  for  k$^{th}$ SGW \\
        $\lambda^{(s)}$ & The mean session arrival  rate  for  s$^{th}$ RIS-M  \\
      $\xi$ &  Angle of UE movement \\
      $\theta$ & Self-blockage angle of UE\\
      $RIS^{i} $ & Indicator of  RIS $i$  being blocked by   obstacles  \\
      $eNB^{j}$ & Indicator of  eNB $j$   being blocked by   obstacles  \\
      $RIS^{\text{self}}$ & Indicator of RIS   self-blockage\\
      $eNB^{\text{self}}$ & Indicator of eNB   self-blockage\\
      $RR$ & RIS reconfiguration\\
      $HO$ & Handover\\
      $p_a$ & Session success rate\\
      $\gamma$ & Signaling rate \Bstrut \\
      \hline
    \end{tabular}
  \end{center}
\end{table}

\subsection{Reference Scenario and Assumptions}
The reference scenario for the analysis is shown in Fig. \ref{fig:scenario_RIS_eNB}. We show a set of eNBs, RISs, a mobile UE,  obstacles and the area of self-blockage.  The following assumptions are made:
\subsubsection{Topology} eNBs are modeled as a homogeneous PPP with eNB density $\lambda_{eNB}$. Each eNB can serve a UE placed in an area of origin $o_{eNB}$ and a radius $d_{eNB}$, which means that the number of eNBs in a disk area follows a PPP distribution. Given a specific number of eNBs in a disk,  we assume that the eNBs locations follow a uniform distribution, thus the distance between any pair of UE and eNB can be obtained as:
\begin{equation}
f_{R_{LoS}}(r) = \frac{2r}{R_{LoS}^{2}} 
\end{equation}
, where $r$ is UE-eNB distance, and $R_{LoS}$ is the radius of LoS. 

We consider that around every eNB, a set of RIS devices are placed in static structures to be used for alternative path creation. For the simplification of the model, we assume that a RIS can only be used by one eNB. We assume that RISs are  distributed in $\mathbb{R}^2$ space also according to the homogeneous Poisson Point Process (PPP) with intensity $\lambda_{\text{RIS}}$, akin to related assumptions  applied to the base stations in \cite{jain2019impact}. Thus, the number of RISs in an area $A$ is a random variable that follows Poisson distribution with mean value $A\lambda_{\text{RIS}}$.

\subsubsection{Obstacle model}  Obstacles are modeled using random shape theory \cite{jain2019impact}, where they have length $l$ and width $w$. When an obstacle has a known position, the blocked area due to the obstacle can be calculated as closed form. When the position is not known \emph{a priori}, the probability  that a UE-RIS or UE-eNB link  is blocked by an obstacle is given by:
\begin{equation}\begin{split}
P(RIS^{i} | r)= &1-e^{-(\beta r + \beta_0)},\\ P(eNB^{j} | r_{eNB})= & 1-e^{-(\beta r_{eNB} + \beta_0)}
\end{split}
\label{eq:staticblockage}
\end{equation}
where $\beta=\frac{2}{\pi}\lambda_B (\mathbb{E}[l] + \mathbb{E}[w])$ and $\beta_0= \lambda_B \mathbb{E}[l]  \mathbb{E}[w]$, with $\lambda_B$ is the density of  obstacles per Km$^2$, $\mathbb{E}[l] $ and $\mathbb{E}[w]$ are the average length and width of  obstacles (e.g., buildings). 

\subsubsection{Self-blockage model}  UE can also block a set of RISs and eNBs, referred to as self-blockage. The zone blocked by the user is a section of a disc with a blocking angle $\theta$.  When we consider a known position and direction of the user, we can deduce the blocked area by the user. The probability that a randomly selected RIS or eNB is blocked is given by 
\begin{equation}
P(RIS^{\text{self}}) = P(eNB^{\text{self}})= \frac{\theta}{2 \pi}
\label{eq:selfblockage}
\end{equation}

\subsubsection{Mobility} As shown in  Fig. \ref{fig:scenario_RIS_eNB}, UE located in L$_1$ is moving to location L$_2$ with a speed of $d_U$ per unit of time, assuming that L$_1$ and L$_2$ are both located in the shaded area where the line of Sight (LoS) from the serving eNB is blocked, and the communication is only possible through RISs. When the user moves to the non-blocked area, and there is a LoS signal from the serving eNB to UE, no RR is required. Here, two cases can be identified: either the serving eNB sends the signal directly to UE, or there is handover to a target eNB. At location L$_1$,  UE is associated to serving RIS with distance $d_{UsR}=r$, and the distance to the serving eNB is $d_{UeNB}=r_{eNB}$. After the movement with distance $d_U$ at time unit and an angle $\xi$, the distance at location L$_2$ becomes equal to $R$, $d_{UtR}=R$, and the distance to the serving eNB is $d_{UeNB}=R_{eNB}$.  We assume that the speed and angle of movement are randomly distributed.  

\subsubsection{RIS reconfigurations} We assume that when there is no direct link between the transmitter $B$ and the receiver $U$ they are connected over a RIS. We furthermore assume that all RISs in the system have the same size and $U$ is always associated to the closest RIS with a direct link. The obstacle blocking the Line-of-Sight (LoS) between $B$ and $U$ is assumed to have predefined location creating blocked area shaded due to  obstacles, as shown in  Fig. \ref{fig:RR_staticBlocking}. We assume that due to the obstacles the signal is attenuated below the SNR threshold, which requires reconfigurations. We consider a disc of a center which is the UE location $o_{s}$ and a radius $r$ as the distance between the serving RIS and the UE before mobility. The UE is considered mobile and moving to a new location where the distance between the UE and serving RIS is $R$. We assume that a RIS reconfiguration  process would happen if there is another RIS in the disc with origin $o_t$ (target UE position) and radius $R_t$ without being blocked by obstacles, i.e., $R_{t} < R$. 
\subsubsection{Handovers (HO)} We consider a disc with the UE location $o_{s}$ as center and a radius $r_{eNB}$ as the distance between the serving eNB and the UE before mobility.   UE moves to a new location with the distance $R_{eNB}$ between the UE and serving eNB. We assume that a handover happens if there is another eNB in the disc $o_t$ (target UE position) and radius $R_{eNB}$ without being blocked by obstacles, i.e $R_{teNB} < R_{seNB}$. In reality, HO  and RR events can be dependent, for instance in the situations when UE mobility triggers a handover and a RR at the same time. For simplification of the analysis and the signaling, we treat these two events independently, and assume that HO events do not include RR. 

\subsubsection{Signaling} The signaling traffic towards the system generally considers the user authentication UEs, and in case of mobility it needs to include the signaling due to HO  and RR process. We assume that the new UE session arrival process from the  k$^{th}$ SGW is Poissonian with a mean rate of $\lambda_{k}$. For the mobility scenario with RR, we assume that that the new UE session arrival process triggering  the  s$^{th}$ RIS-M is also Poissonian with a mean rate of $\lambda^{(s)}$.

\subsection{Handover (HO) rate }
A handover happens when at least one eNB is closer to the UE than the serving eNB.
\begin{thm}\label{thm:HO}
The HO rate can be expressed as the sum over all SGW of the product of the session arrival rate $ \lambda_{k}$ for the k$^{th}$ SGW and the  probability of HO when a UE moves with a speed $d_{U}$ and a angle of movement $\xi_{eNB}$, given as follows:
\begin{equation}\label{eq:HOR_general}\begin{split}
\mathbb{E}[HO] = \sum_{k=1}^{N_{SGW}}     P(HO| d_U,\xi_{eNB}) \lambda_{k}, \text{ and } \\P(HO| d_U,\xi_{eNB})  =  1 - e^{   -P(Z)   \lambda_{eNB} \pi R_{eNB}^{2}  } 
\end{split}
\end{equation} 
where $R_{eNB}$ can be computed as shown in Appendix \ref{appedix:proof1} (see Eq. (\ref{eq:R_from_r})) in relation to the speed of the UE $d_{U}$, the angle of movement $\xi_{eNB}$ and the distance between the UE and the serving eNB $r_{eNB}$ as follows:
\begin{equation}
R_{eNB}= \sqrt{r_{eNB}^2 + d_{U}^2 - 2r_{eNB} d_{U}\cos(\pi - \xi_{eNB})}
\end{equation}
and   the probability $P(Z)$ that an eNB   is not blocked by  obstacles or self-blockage, given as:
\begin{equation}\begin{split}
P(Z) = (1-\frac{\theta}{2\pi})  \frac{2 e^{-\beta_0}}{\beta^{2}R_{LoS}^{2}}\left[ 1- (1+ \beta R_{LoS} ) e^{-\beta R_{LoS}} \right]
\end{split}\label{eq:PZ}
\end{equation} 
\end{thm}

\begin{proof} 
See Appendix \ref{appedix:proof_th}.
\end{proof}

\subsection{RIS Reconfiguration (RR) rate}
To obtain the RIS Reconfiguration  (RR) rate, we consider two cases: (i)  obstacles known \emph{a priori}, and ii)  obstacles unknown. \\
\subsubsection{Obstacles locations known \emph{a priori}}

\begin{figure}
 \centering 
   \includegraphics[scale=0.25]{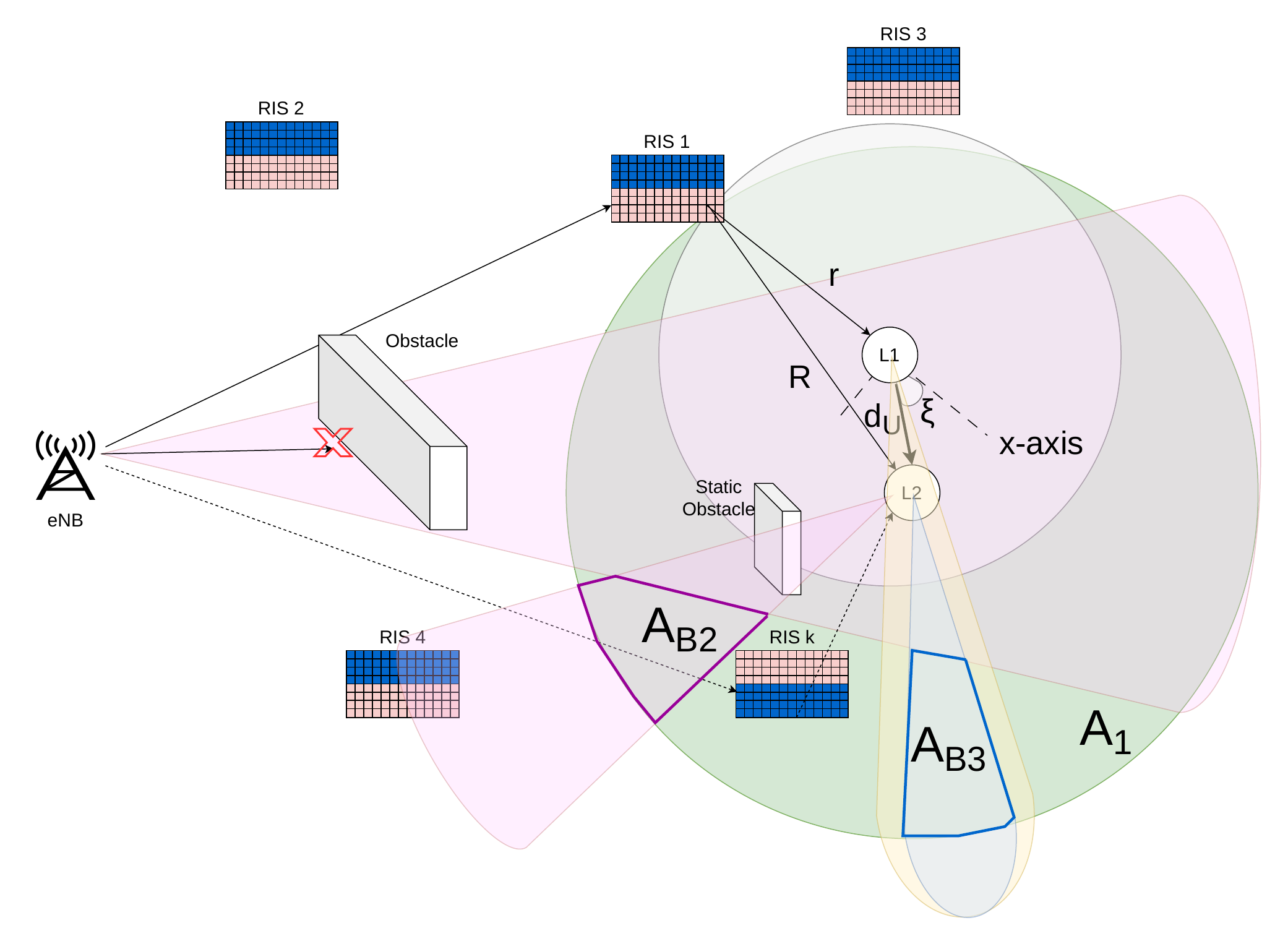}
 \caption{Scenario with mobile UEs, RISs,  obstacles and self-blockage. }
\label{fig:RR_staticBlocking}
\vspace{-0.5 cm}
\end{figure}
The RR rate can be expressed as the sum over all RIS-M of the product of RR probability and the mean session arrival rate for s$^{th}$ RIS-M. 
\begin{equation}\label{eq:RRP_general}
\mathbb{E}[RR] =  \sum_{s=1}^{N_{RISM}}  P(RR) \lambda^{(s)} 
\end{equation} 
The RR probability can be expressed as the probability of RIS reconfiguration when UE  $U$ moves from L$_1$ to L$_2$ as follows:
\begin{equation}\label{eq:RRR_general}
 P(RR)=1-e^{-A\lambda_{\text{RIS}}}
\end{equation} 
where $A$ is the non-blocked excess area in Fig. \ref{fig:RR_staticBlocking}. The new closer RIS  cannot be in the circle of radius $r$ as we assumed that RIS 1 is the closest to user in location L$_1$, also the selected  RIS candidate cannot be in the  blocked area. As a result, the non-blocked excess area  is the area inside the circle with radius $R$ is deprived of the intersection between the circle with radius $r$ and the blocked area. \\

\begin{prop}\label{prop:RRR_no_obstacle}
The RIS reconfiguration probability for a mobile UE with a speed of $d_U$ at unit of time at an angle $\xi$ of a RIS-assisted system in presence of obstacle between the RIS and the UE, where the distribution of RISs follows the PPP distribution with density $\lambda_{RIS}$ is given as follows:
\begin{equation}\label{eq:RRR_no_extra_obstacle}
P(RR| d_U,\xi)=1-e^{-A_1 \lambda_{\text{RIS}}}
\end{equation} 
where $A_1$ can be determined in terms of the angle resulted from the obstacle, the distance between the base station $B$ and the UE  $U$, the speed $d_U$ at unit of time and the angle $\xi$.
\end{prop}
\begin{proof}
See Appendix \ref{appedix:proof1}, where all the cases of finding $A_1$  as function of the  blocked area are solved. 
\end{proof}

\textbf{RR with  obstacle to the RIS:}
In Fig. \ref{fig:RR_staticBlocking},  an obstacle can exist in the visible excess area defined previously as $A_1$. A static object outside of the area defined by the circle of origin L$_2$ and radius $R$ is not considered because a RR happens only when a RIS is closer to the user than the operating RIS, and those RISs can be blocked by an object between the user and the RIS. The  obstacle  creates a new blocked area denoted as $A_{B_2}$ and might decrease the  visible excess area $A_1$ to obtain a visible excess area $A_2$. Thus $A_2$ is given as follows:
\begin{equation}
A_2=A_1 - A_{B_2}
\end{equation} 
where $ A_{B_2}$ can be computed numerically based on the location and the size of the obstacle, which are assumed to be known for a static object in an indoor or outdoor environment. \\
Thus, the RR probability for a mobile UE with a speed of $d_U$ at unit of time at an angle $\xi$  in presence of obstacle between the RIS and the UE, and an obstacle in the space, can be calculated using  Eq. (\ref{eq:RRR_general}) by replacing $A$ with the obtained area $A_2$, i.e., 
\begin{equation}\label{eq:static-blockage}
P(RR| d_U,\xi)=1-e^{-(A_1- A_{B_2})\lambda_{\text{RIS}}}
\end{equation}

\textit{\textbf{RR with self blockage: }} 
 Self-blockage happens when the body of the user creates a blocked zone as a sector of a circle with origin the location of the user, and a predefined blocked angle $\theta$ depending of the user shape. In Fig. \ref{fig:areas2}, we illustrate different scenarios of self-blockage based on the blockage direction and angle. We assume that any high-frequency signal coming from a RIS or relay in the blocked area is totally blocked by the body of the mobile user, thus for RR, only devices, RIS, in the visible area are considered, which are marked with triangles  in the figure for different cases. In case 1, the self-blocked area is entirely inside the invisible area blocked by the obstacle placed between the base station $B$ nd the user, which means that that the  visible excess area $A_1$ remains unchanged, and RR is given by Eq. (\ref{eq:RRR_no_extra_obstacle}). This case can be derived from self-blocked angle, the location of the mobile user L$_2$, the speed $d_U$ per unit of time, and the angle $\xi$, and the blocked area. In cases 2,3 and 4, part of the visible excess area $A_1$ is covered by the self-blockage. The difference appears in the direction of the self-blocked  in terms of link between the user and the operating RIS, which affects the calculation of the extra blocked area denoted as $A_{B_3}$. In this work, the blocked area $A_{B_3}$ can be calculated numerically, because all the geometrical information about the obstacle are given in advance; the self-blocked angle $\theta$ and the direction of movement defined by the angle $\xi$. Thus
the  visible excess area $A_3$ and the RR rate for a mobile user with a speed of $d_U$ at unit of time at an angle $\xi$ and a self-blocking angle $\theta$,   are given as follows:
\begin{equation}\label{eq:self-blockage}
A_3=A_1 - A_{B_3}, \;\;\;\;P(RR| d_U,\xi,  \theta)=1-e^{-A_3 \lambda_{\text{RIS}}}
\end{equation} 
\begin{figure*}
 \centering 
   \includegraphics[scale=0.12]{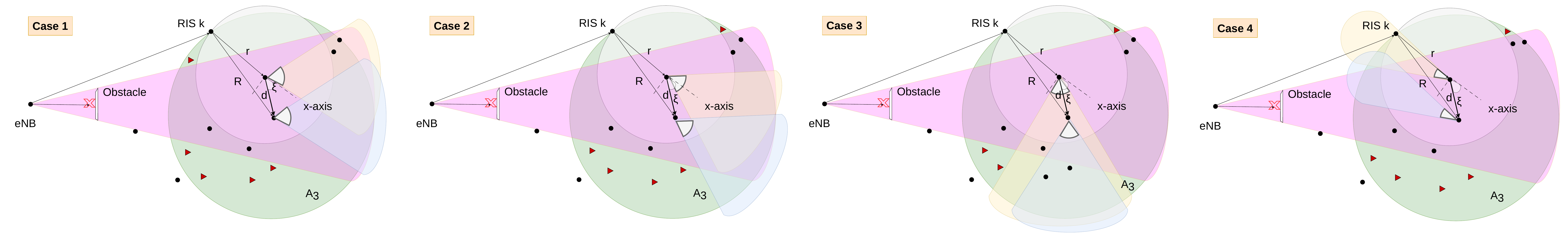}
 \caption{Blocked area situations based on self-blockage direction. }
\label{fig:areas2}
\vspace{-0.5 cm}
\end{figure*}
 

 \textbf{RR with random speed and direction mobility: } Let us consider a user that moves at a randomly distributed speed $d_U\in [d_U^{min},d_U^{max}]$ per unit of time at a determined angle $\xi$. Let $D$ be a random variable that defines the speed of the user per unit of time. If we denote the pdf of $D$ as $f_D (d_U)$, the RR probability  can be obtained as follows:
 \begin{equation}
P(RR | \xi)=1-\int_{d_U^{min}}^{d_U^{max}} f_{D}(d_U) e^{-\lambda_{\text{RIS}} A}\partial d_U
\end{equation}
In general, the speed of the mobile user can follow different distributions which depends on user behavior. Let us assume that the user speed is uniformly distributed in the interval $d_U\in [d_U^{min},d_U^{max}]$. The pdf of $D$ and the RR probability can be obtained as follows:
\begin{equation}\begin{split}
f_D (d_U)=  &\frac{1}{d_U^{max} - d_U^{min}}\\ \Rightarrow  P(RR | \xi)= &1-  \frac{1}{d_U^{max} - d_U^{min}} \int_{d_U^{min}}^{d_U^{max}}  e^{-\lambda_{\text{RIS}} A}\partial d_U
\end{split}
\end{equation}

\textit{\textbf{RR with mobility at random angle: }}
This case considers a random waypoint model, where a user can move in a  speed $d_U$ per unit of time at a randomly distributed angle  $\xi\in [\xi^{min},\xi^{max}]$. Let $X$ be a random variable that defines the angle of movement  of the user, and $f_X (\xi)$ is the pdf of $X$.  Therefore the RR at a random angle $\xi$ can be obtained as follows:
 \begin{equation}
P(RR | d_U )=1-\int_{\xi^{min}}^{\xi^{max}} f_{X}(\xi) e^{-\lambda_{\text{RIS}} A}\partial \xi
\end{equation}
A particular case can  simplify the RR expression is when the angle $\xi$ is uniformally distributed on $[0, \frac{\pi}{2} ]$,  given as follows:
\begin{equation}
 P(RR | \xi)=1-  \frac{2}{\pi} \int_{0}^{\frac{\pi}{2}}  e^{-\lambda_{\text{RIS}} A}\partial \xi
\end{equation}

\textbf{RR with mobility at random speed and random angle (random waypoint and random direction): }
When the user is moving at randomly distributed speed $d_U\in [d_U^{min},d_U^{max}]$  per unit of time at a randomly distributed angle  $\xi\in [\xi^{min},\xi^{max}]$,  RR can be calculated as:
  \begin{equation}
P(RR)=1-\int_{\xi^{min}}^{\xi^{max}} \int_{d_U^{min}}^{d_U^{max}} f_{D}(d_U) f_{X}(\xi) e^{-\lambda_{\text{RIS}} A}\partial d_U \partial \xi
\end{equation}
When both the speed $d_U$ per unit of time the angle $\xi$ are  uniformally distributed on $[d_U^{min},d_U^{max}]$ and $[0, \frac{\pi}{2} ]$,  respectively,  RR calculation can be simplified as follows:
 \begin{equation}
P(RR)=1-\frac{2}{\pi (d_U^{max} - d_U^{min})} \int_{0}^{\frac{\pi}{2}} \int_{d_U^{min}}^{d_U^{max}} e^{-\lambda_{\text{RIS}} A}\partial d_U \partial \xi
\end{equation}

\subsubsection{Obstacles locations are not known \emph{a priori}} In this case, the   visible excess area to the mobile UE  cannot be expressed deterministically. We consider that RR occurs when at least one RIS is located in the visible excess area to the mobile UE, and RISs are distributed in $\mathbb{R}^2$ with density $\lambda_{RIS}$.  The RR rate can then be presented as the sum over all SGW of the product of session arrival rate $\lambda_{k}$ and  the probability of RR when a UE moves with a speed $d_U$ and angle of movement $\xi$. Using Theorem \ref{thm:HO} and replacing $\lambda_{eNB}$, $R_{eNB}$, and $ \xi_{eNB}$ by $\lambda_{RIS}$,   $R_{RIS}$, and $\xi$, we obtain the RR rate as follows:
\begin{equation}\label{eq:RRR_general_last}\begin{split}
\mathbb{E}[RR] =& \sum_{s=1}^{N_{RISM}}P(RR| d_U,\xi) \lambda^{(s)}  \\=&  \sum_{s=1}^{N_{RISM}} \lambda^{(s)}  \left( 1 - e^{   -P(Z)   \lambda_{RIS} \pi R_{RIS}^{2}  } \right)
\end{split}
\end{equation} 
where $P(Z)$ is defined in Eq. (\ref{eq:PZ}) and $R_{RIS}$ is given as follows:
\begin{equation}
R_{RIS}= \sqrt{r_{RIS}^2 + d_{U}^2 - 2r_{RIS} d_{U}\cos(\pi - \xi)}
\end{equation}
The proof of this expression can be found akin to Theorem \ref{thm:HO}.

\subsection{Signaling Rate} 

The expected value of the signaling rate denoted as $\gamma$ is obtained as the aggregate rate of the expected values of basic signaling $S_b$ (e.g., keep alive during the user session lifetime) and of the additional signaling  (overhead, $S_o$) which incorporates the expected values of HO and RR rates, denoted as $\mathbb{E}[Sb]$, $\mathbb{E}[So]$, $\mathbb{E}[HO]$ and $\mathbb{E}[RR]$, respectively, i.e., 
\begin{equation}\label{eq:sig_rate}
\mathbb{E}[\gamma]= \mathbb{E}[Sb] + p_a \mathbb{E}[So] 
\end{equation}
where $p_a$ is the session success rate (i.e., which is always equal 1 if the session does not fail).\\

For every session, and if the UE does not change its serving eNB  or RIS during the session, we only have basic signaling. The expected value of the basic signaling rate can be seen as a compound Poisson process of the random number of messages
sent during the mobile gateway holding times and RIS-manager holding times, i.e., 
\begin{equation}\label{eq:Au_fix}
\mathbb{E}[Sb] = \sum_{k=1}^{N_{SGW}}\lambda_{k}+ \sum_{s=1}^{N_{RISM}}\lambda^{(s)}
\end{equation}

We now extend Eq.(\ref{eq:Au_fix}) to include the effect of mobility. Two events can occur during a session holding time, i.e.,  HO and RR, which in every networking management system are to be counted as additional signaling (overhead). Both processes occur at the corresponding rates of RR and HO. Staring with the RR, we design a series of signaling messages according to the current HO standard, including RR requests, admission control, etc. For simplicity, we assume this process never to fail, and be deterministic in its performance. Similarly, the HO signaling (here designed according to standard \cite{NR2019}) includes a number of signaling messages that engage SGW and MME.  Also all of these messages occur at the rate of HO. Thus, the expected value of the resulting signaling rate of this overhead can be defined as:
\begin{equation}\label{eq:Au_mobile}
\mathbb{E}[So] = \sum_{k=1}^{N_{SGW}}  \mathbb{E}[HO]    + \sum_{s=1}^{N_{RISM}}  \mathbb{E}[RR]
\end{equation}

By substituting  Eq. (\ref{eq:Au_fix}) and Eq. (\ref{eq:Au_mobile}) in Eq. (\ref{eq:sig_rate}), we obtain the expression of signaling rate as follows:
\begin{equation}\label{eq:sig_rate_final}\begin{split}
\mathbb{E}[\gamma]= & \left[ \sum_{k=1}^{N_{SGW}} \left(1 + P[HO]  \right) \lambda_{k}  +   \sum_{s=1}^{N_{RISM}}  (1+P[RR])\lambda^{(s)} \right] \\
 & \times \left(1  + p_{a} \right)\\
\end{split}
\end{equation}

\subsection{Discussion on signaling protocol} \label{sec:signaling}
\begin{figure*}
 \centering 
   \includegraphics[scale=0.6]{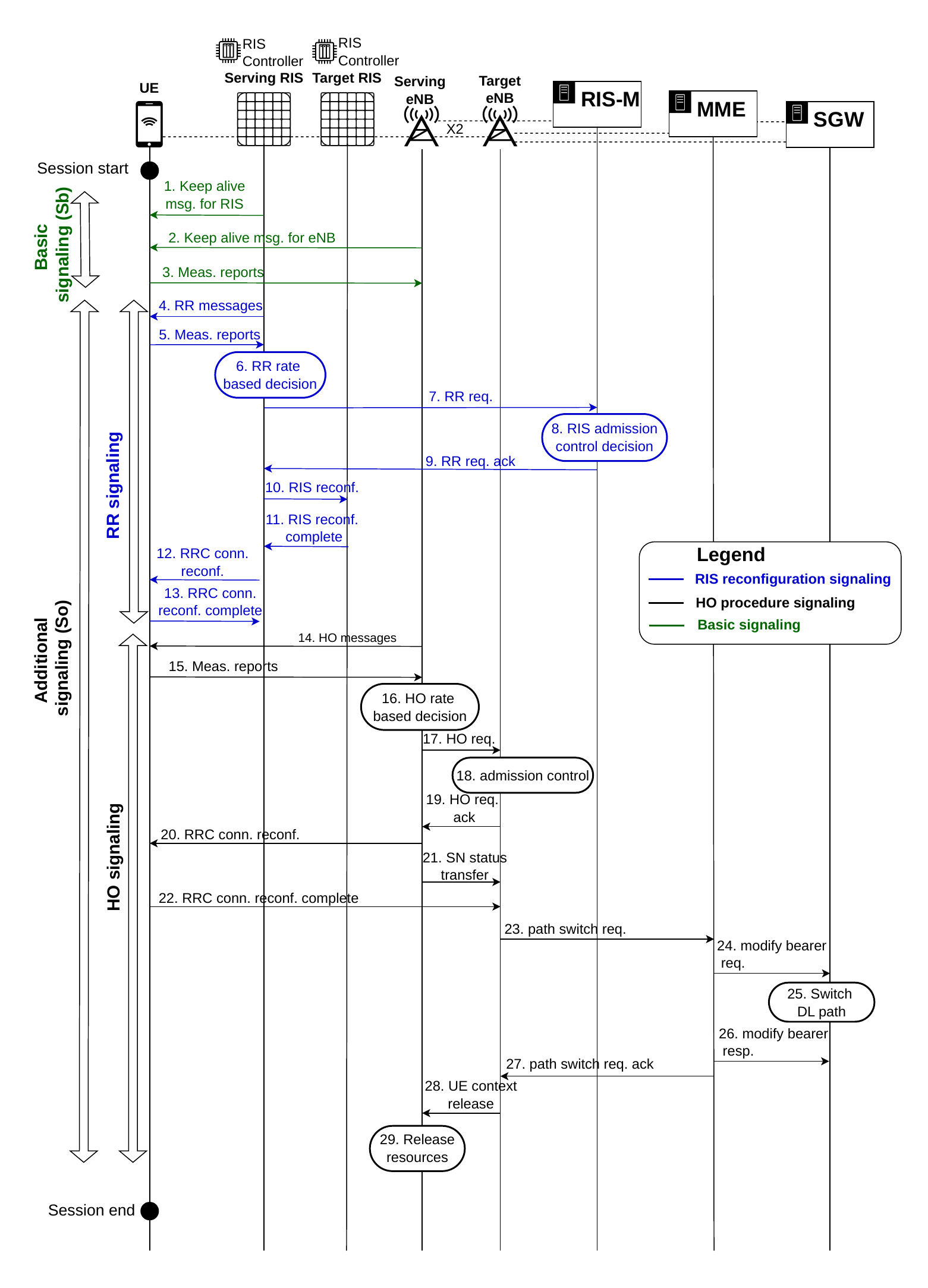}
 \caption{L3 signalling for RIS reconfiguration and HO procedure. }
\label{fig:RIS_HO_RRR_proc}
\vspace{-0.5 cm}
\end{figure*}
%

The signaling rate is only a statistical property of the rate of the arrival of RIS reconfigurations and handovers. It is a useful information which can be used to dimension network management systems. In real networks, the L3 signaling for RIS reconfiguration require a decision making unit, akin to the decision making on imminent handovers. The network management system also needs to select which RIS device to allocate, based on various criteria, such as distance to the device, received signal strength (RSS) parameters, or based on some predefined thresholds. For instance, the distance-based approach, which we used in our model measures the distance between the UE and each RIS device and eNB, and selects the RIS that provides the shortest distance. Similarly, RSS based approach would choose the path with highest RSS. 

While the amount of signaling determines the size of the control servers in the network management system, the signaling protocol itself needs to be designed such that the minimum delay is experiences when reconfiguring RIS, similar to handover. Especially for high speed connections, it is of major importance to design a signaling protocol which does not experience significant delay during the reconfiguration process. To this end, the existing handover signaling protocols can be repurposed for RIS reconfigurations.  Figure  \ref{fig:RIS_HO_RRR_proc} details one such example based on L3 signaling protocols. In the basic signaling part, the serving eNB starts the UE measurment procedure by sending control messages through different paths (Step 2). Similarly, the serving RIS can start such messaging, and in case it is designed as passive, the base station can do it on its behalf (Steps 1 and 4). The UE sends the measured reports to the serving eNB and RIS (Steps 3 and 5).
Using the measured reports the serving eNB makes a decision about RIS reconfiguration (RR), and this happens at RR rate precisely (Step 6). The serving RIS/eNB sends RR request message to RIS-M to prepare the RIS reconfiguration (Step 7), where the message contains the needed information. The RIS-M executes a RIS admission control to check if the target RIS can be used to create a new path (Step 8).  The RIS-M sends a RR request acknowledgment to the serving eNB (Step 9). The serving RIS/eNB generates and sends a RIS configuration message to the target RIS with the needed information to reconfigure the RIS (Step 10). Those messages will be read and compiled by the RIS controller of each RIS device. Once the  target RIS has successfully completed the reconfiguration, the RIS controller sends the RIS reconfiguration complete messages to the sending RIS (Step 11). The serving RIS generates and sends Radio Resource Control (RRC) message to UE with the information about the RIS reconfiguration (Step 12).  When UE successfully access the serving eNB through the new RIS configuration, it sends the RRC connection reconfiguration complete message to the serving RIS (Step 13). 

The RIS reconfiguration signaling shown here is a novel proposal for signaling and is based on the similar procedure for the HO signaling (Steps 14-29), also shown in Figure  \ref{fig:RIS_HO_RRR_proc}. It should be noted that both HO and RR rates are here illustrated only to show their major impact on the dimensioning in the network control and management infrastructure. Further optimizations are possible  and needed (for instance by off loading all signaling from RISs to the base stations) but this is outside the scope of this paper. 

\vspace{-0.2cm}



\section{Numerical Evaluation}\label{sec:results}

\subsection{Probility of RIS reconfigurations}
To analyze the RIS reconfigurations probability, we first focus on the scenarios where obstacles are known \emph{a priori} in a squared room of size $10\times 10$ m$^2$. Using the model illustrated earlier in Figure \ref{fig:RR_staticBlocking}, we assume that there is one eNB,  while the signal to the  mobile UE  might also suffer from self-blockage. Our theoretical results using  stochastic geometry are here also verified by Monte Carlo simulations implemented in Python. In the  simulation, we run $Z=10^5$ independent executions to obtain the average RR probability. We generate the coordinates of RIS devices using the PPP distribution in an area of the defined room, the base station is considered fixed in the point $(0, 4 )$, an obstacle that might represent a wall is defined between the two points $(4, 2.8)$ and $(4,5)$, which results an invisible area between the mobile UE and eNB. We set an obstacle defined by the segments between the two points $(5,2)$ and $(6.5,2)$, in order to analyze the impact of an obstacle. The main  parameters used for both the numerical evaluation and for simulation are defined in Table \ref{tab:table1}. 

\begin{table}[h!]
  \begin{center}
    \caption{Parameters used in the scenario with known obstacles}
    \label{tab:table1}
    \begin{tabular}{|c|c|c|}
    \hline 
      \textbf{Parameters} & \textbf{Static values}&\textbf{Random  values}  \\
      \hline
      $n$ & 1 & \\ \hline
       $\lambda_{RIS}$ & 0.2 & Uniform$(0,1)$ (device/m$^2$) \\
      \hline
       $\theta$ &40$^{\circ}$ & 45$^{\circ}$, 90$^{\circ}$, 120$^{\circ}$ \\ \hline
      $\xi$ & 45$^{\circ}$ & Uniform$(0^{\circ}, 90^{\circ})$ \\\hline
      r$_{RIS}$& 2 m &  \\ \hline
       d$_U$ & 2 (m/s) & Uniform$(1, 2.5)$ (m/s) \\ 
      \hline
    \end{tabular}
  \end{center}
\end{table}

Figure \ref{fig:RRR_density_det} and \ref{fig:RRR_density_unif_self}  illustrate the RIS reconfiguration probability in relation to different  RIS densities and user speed under different setups and in  three defined scenarios: (i) no obstacles using Eq. (\ref{eq:RRR_no_extra_obstacle}),  (ii)  obstacles using Eq. (\ref{eq:static-blockage}), and (iii) self-blockage with blocking angle $\theta=40^{\circ}$ using Eq. (\ref{eq:self-blockage}). The RIS density $\lambda_{RIS}$ varies from $0.001$ to $1$ device per square meter.  The results show that the RIS reconfiguration probability increases when the density of devices increases, which means that there is a higher probability to find a closer RIS  than the operating RIS  in the visible area (Figure \ref{subfig:1}).  Both the  obstacle and the self-blocking affects the RIS reconfiguration probability, which can be explained by the decrease of the visible excess area due to blockage, and thus the decreased chance of finding a closer RIS. Here, the  obstacle has a higher impact on the RR probability than a self-blocking angle of 40$^{\circ}$. Furthermore, the results indicate that a proper RIS  density can be found at which it is almost certain that a reconfiguration will occur. In Figure \ref{subfig:1}, at a density $\lambda_{RIS}=0.4$  RR probability converges to $1$ for a scenario without obstacle, while it converges at $0.7$ and $0.8$ for the self-blocking and  blockage scenarios, respectively. For density of $0.1$ device/m$^2$, RR probability is equal to $0.45, 0.5$ and $0.67$ for the scenarios without obstacles, with  obstacle, and with self-blockage, respectively.  The theoretical results coincide with the results from simulation, which validates the analysis.

\begin{figure*}
 \centering 
 \begin{subfigure}[t]{0.329\linewidth}
   \includegraphics[scale=0.36]{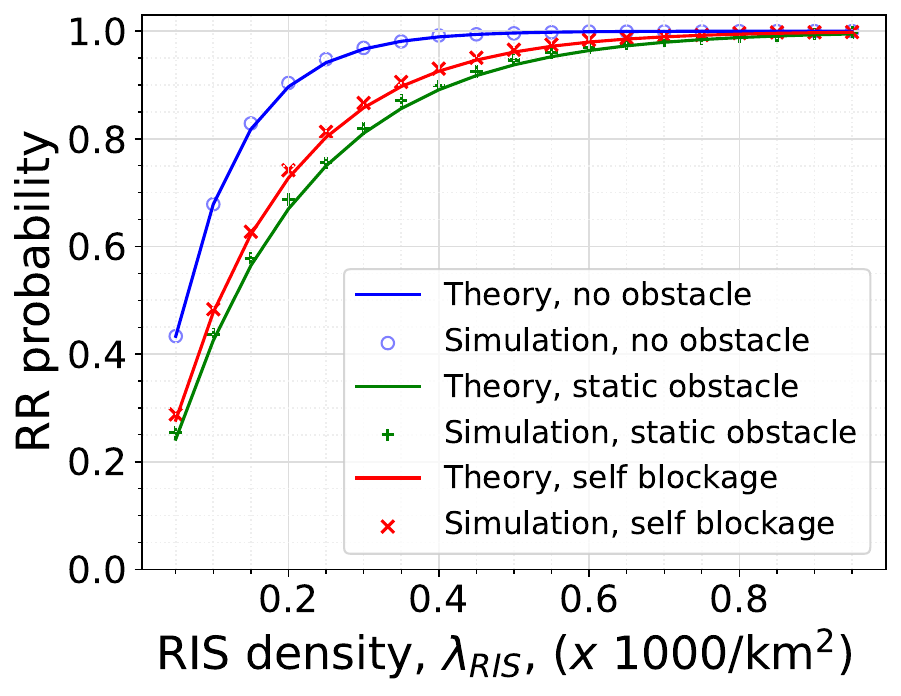}
   \caption{$\lambda_{RIS}$ variable, $d_{U}$ static}\label{subfig:1}
   \end{subfigure}
    \begin{subfigure}[t]{0.329\linewidth}
   \includegraphics[scale=0.36]{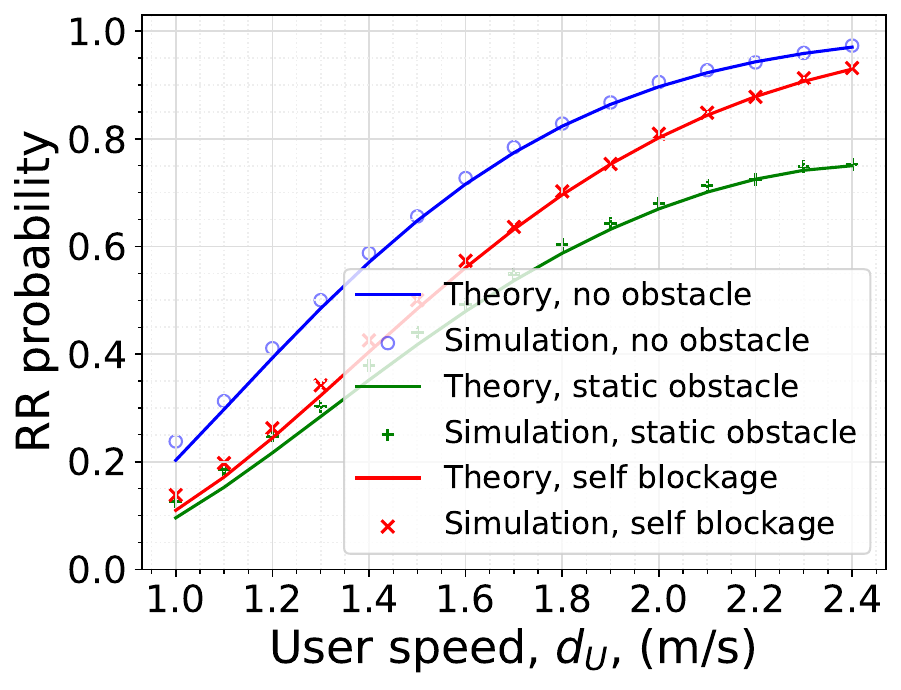}
   \caption{$d_{U}$ variable }\label{subfig:2}
   \end{subfigure}
    \begin{subfigure}[t]{0.329\linewidth}
    \includegraphics[scale=0.36]{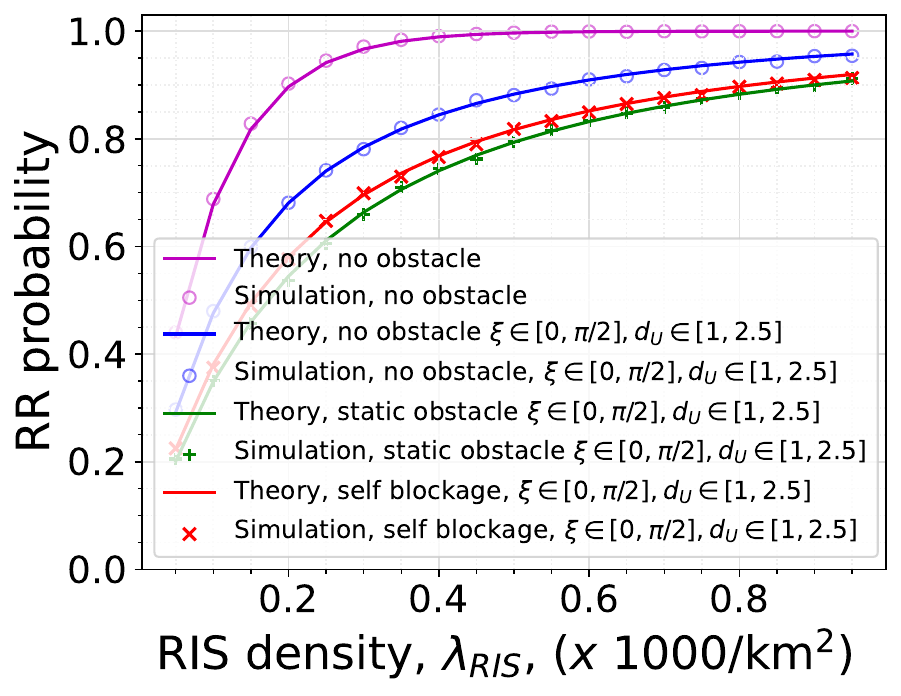}\caption{$\lambda_{RIS}$ variable, $d_{U}$ random}\label{subfig:3}
    \end{subfigure}
 \caption{RR probability vs. RIS density,  static user speed, and RIS density with uniformly distributed angle and user speed}
\label{fig:RRR_density_det}
\vspace{-0.5 cm}
\end{figure*}

\begin{figure}
 \centering 
   \includegraphics[scale=0.4]{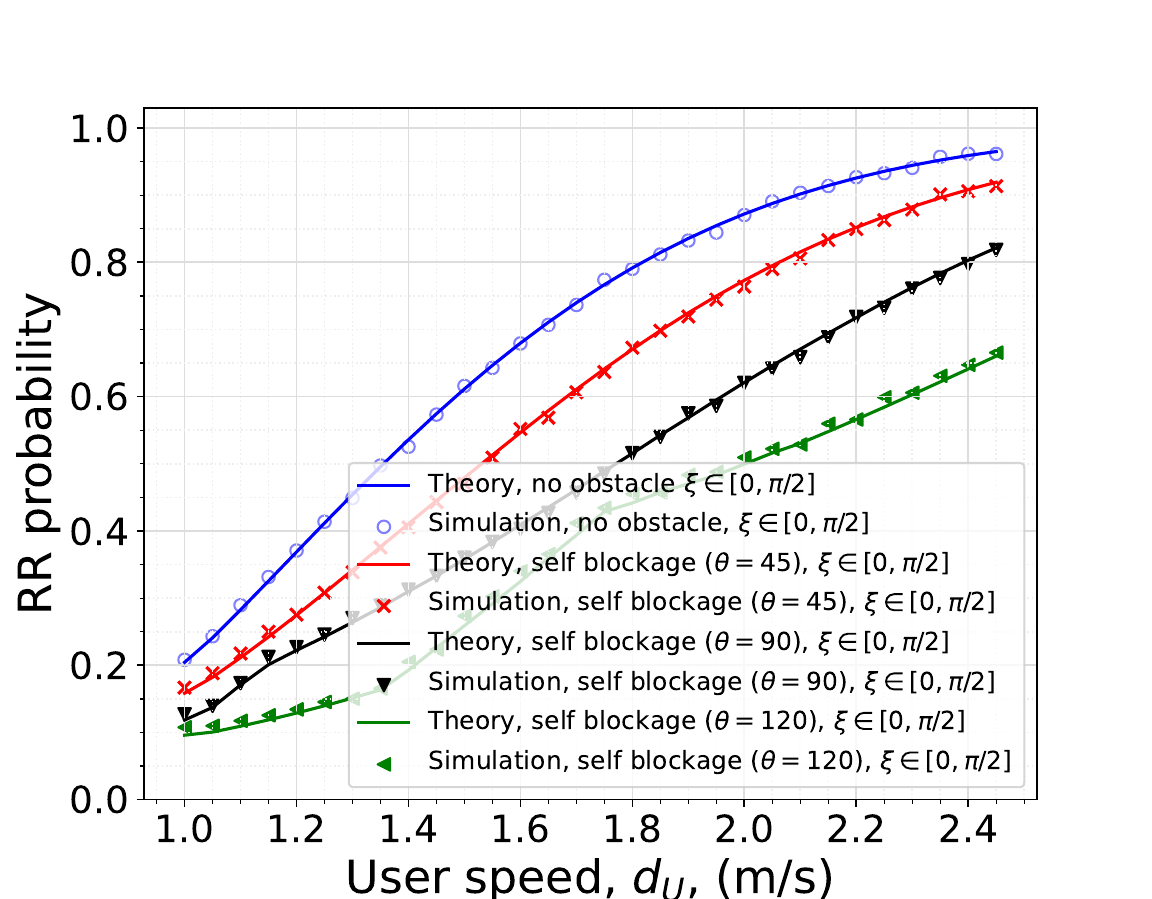}
 \caption{RR probability vs. static user speed with  uniformly distributed deviation angle $\xi$ and different self-blockage angle $\theta$. }
\label{fig:RRR_density_unif_self}
\vspace{-0.5 cm}
\end{figure}

\begin{figure*}
 \centering 
 \begin{subfigure}[t]{0.24\linewidth}
   \includegraphics[scale=0.29]{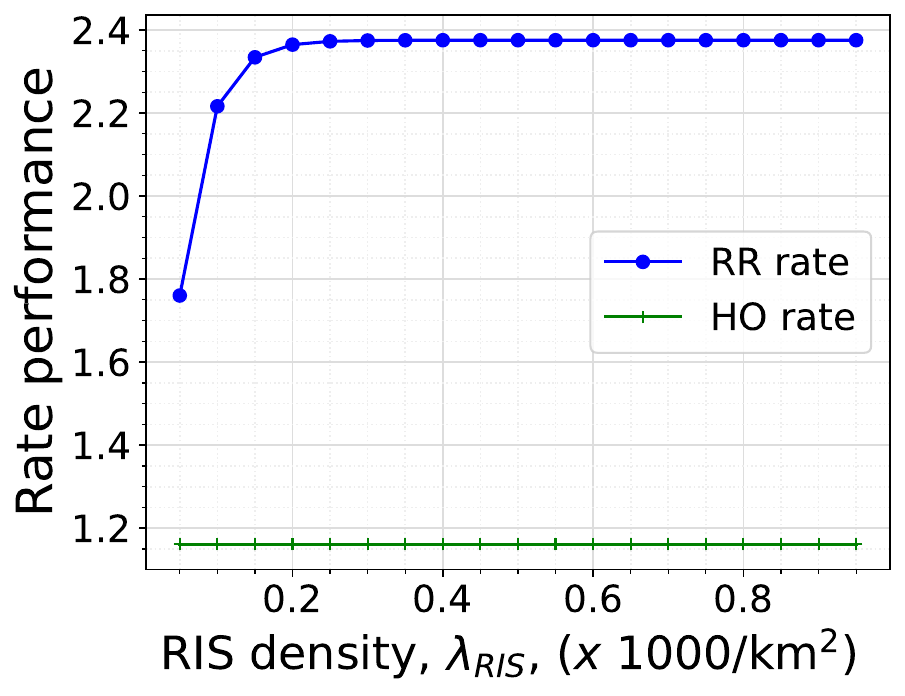}
   \caption{$\lambda_{RIS}$ }\label{subfig:RRHO1}
   \end{subfigure}
   \begin{subfigure}[t]{0.24\linewidth}
   \includegraphics[scale=0.29]{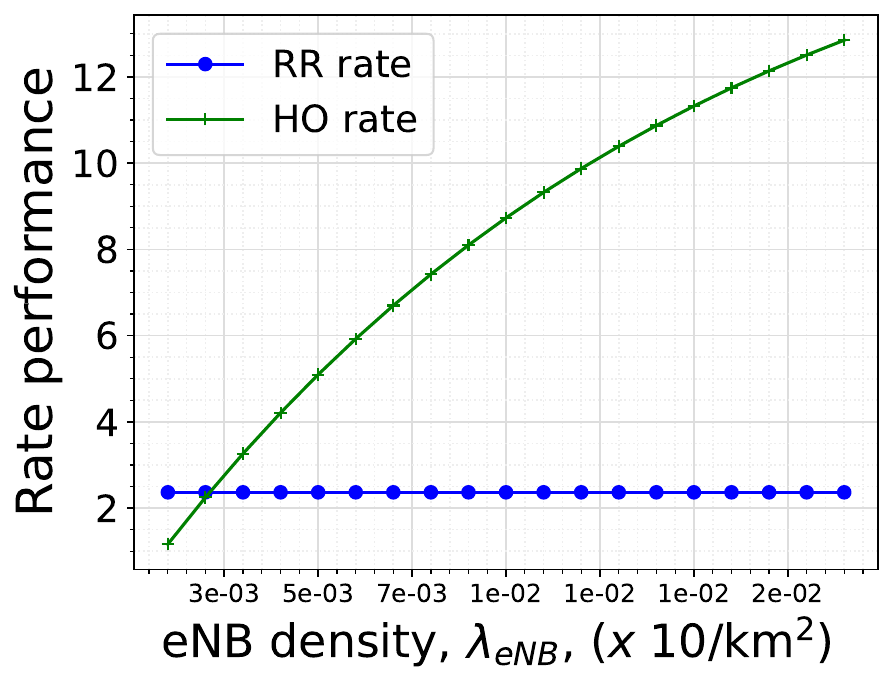}
    \caption{$\lambda_{eNB}$ }\label{subfig:RRHO2}
   \end{subfigure}
   \begin{subfigure}[t]{0.24\linewidth}
   \includegraphics[scale=0.29]{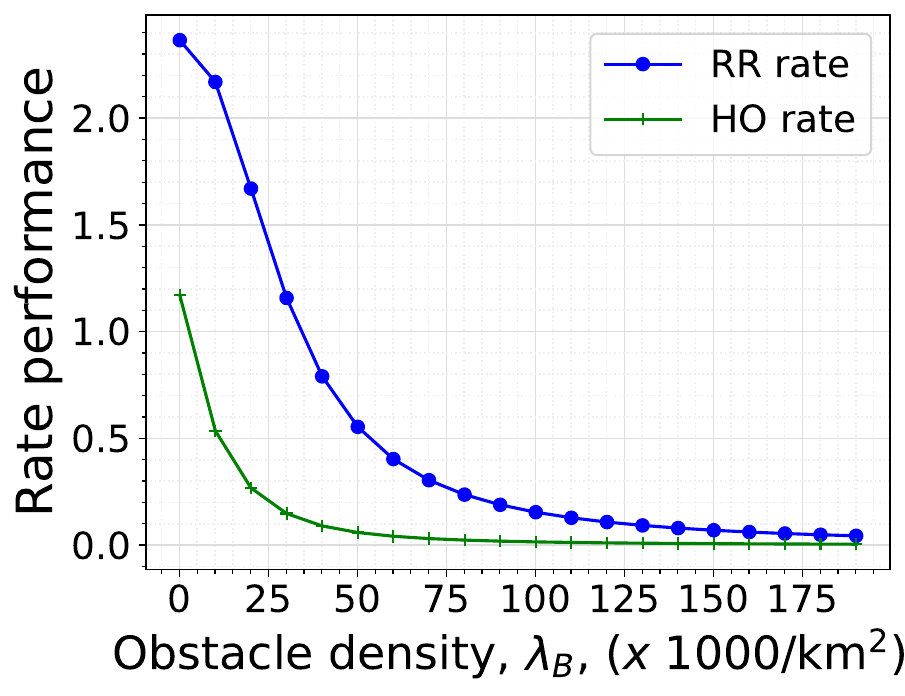}
    \caption{$\lambda_{B}$ }\label{subfig:RRHO3}
   \end{subfigure}
   \begin{subfigure}[t]{0.24\linewidth}
   \includegraphics[scale=0.29]{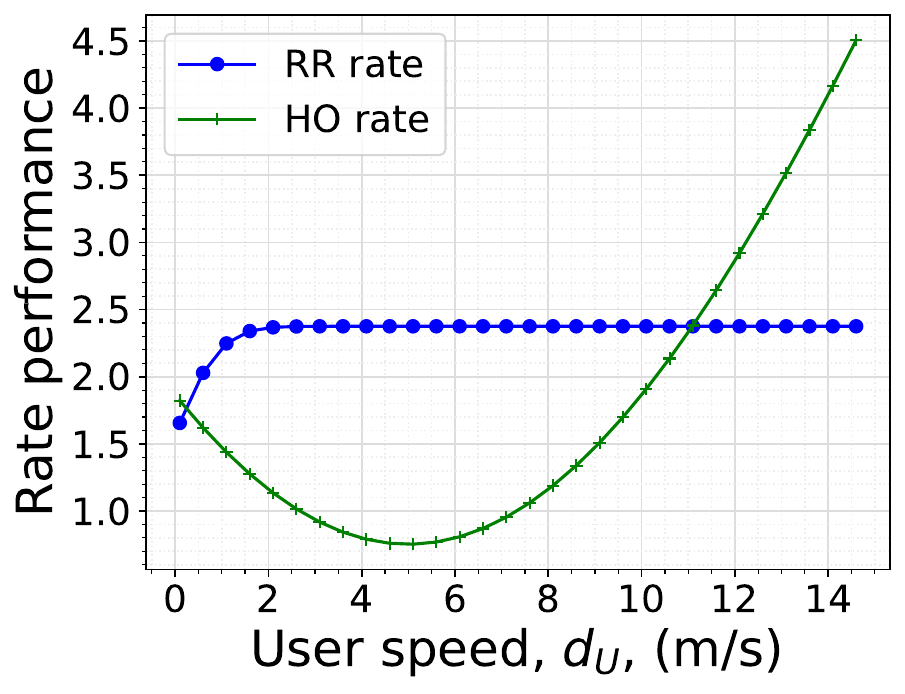}
    \caption{$d_{U}$ }\label{subfig:RRHO4}
   \end{subfigure}
 \caption{RR and HO rates in terms of RIS density, eNB density, obstacle density and  static user speed. }
\label{fig:RRvsHO}
\vspace{-0.2 cm}
\end{figure*}

 Figure \ref{subfig:2}  shows the RIS reconfiguration probability in relation to the speed of mobile user. The deviation angle is considered static and set to $\xi=40^{\circ}$ from the operating RIS. The speed values are chosen to represent a human speed, and are bounded by the size of the invisible area where the eNB cannot communicate directly to the user.   RR probability increases with the user speed in all cases. The impact of static blockage and self-blockage is small at lower speeds, which means that the area blocked by the obstacle is smaller as compared to the visible excess area. With the increasing speed, the impact of  static blockage increases, which can be explained by the fact the user is moving in the direction of the obstacle and thus the blocked excess area becomes larger; thus the probability to find a closer RIS decreases in comparison to the scenario with no obstacles and with self-blockage.

To analyze the impact of blockages on RR probability  with users moving at a random speed, we consider a uniformly distributed speed and  a static angle $\xi=40^{\circ}$.   Figure \ref{subfig:3}  shows the results for RR probability in relation to the RIS density, where RR probability is higher with higher density similar to Figure \ref{subfig:1}. Since the speed is uniformly distributed, RR probability  is lower than with the static speeds, and the obstacles have a higher impact. For density of $0.1$ device/m$^2$, RR probability  is equal to $0.3, 0.4$ and $0.5$ for scenarios without obstacles, with  obstacles, and with self-blockage, respectively. These results can be explained by the fact that the visible excess area when the speed is uniformly distributed is lower in average than the area obtained with static speeds, which implies a lower probability of having a RIS device available in the area. 


%

 Figure \ref{fig:RRR_density_unif_self}  shows a comparison between  RR probability  with different self-blocking angle $\theta=\{45, 90, 120\}$ and with no obstacles and with an obstacle. The user moves at static speeds per unit of time and a uniformly distributed angle of direction $\xi\in [0, \pi/2]$.  We find that the larger the self-blocking angle the higher the impact on RR probability, which is expected as the area blocked increases. The impact becomes higher than the impact of the  obstacle for angles $90^{\circ}$ and $120^{\circ}$. 

\subsection{HO, RR and signaling rates}
The parameters used in this numerical study are defined in Table \ref{tab:table2}. The main assumption here is that the obstacles are not known \emph{a priori}.
\begin{table}[h!]
  \begin{center}
    \caption{Parameters in the scenario with unknown obstacles}
    \label{tab:table2}
    \begin{tabular}{|c|c|c|c|}
    \hline 
      \textbf{Parameters} & \textbf{Values} &   \textbf{Parameters} & \textbf{Values} \\
      \hline
       $\lambda_{RIS}$ & 0.2 & $N_{SGW}$ & 1 \\
      \hline
       $\lambda_{eNB}$ & 0.001 & $N_{RISM}$  & 1 \\
      \hline
       $\lambda_{B}$ & 0.2 &  $R_{LoS}$& 10 km \\
      \hline
       $\theta$ &45$^{\circ}$  & ($l,w$) & ($10 $ m$, 10$ m) \\ \hline
       $\xi$ & 45$^{\circ}$ & $p_a$& 0.99 \\\hline
       r$_{RIS}$& 2 m  & d$_U$ & 2 (m/s)  \\ \hline
    \end{tabular}
  \end{center}
\end{table}

In Figure \ref{subfig:RRHO1} we vary the RIS density $\lambda_{RIS}$, which shows that HO rate, as calculated with Eq. (\ref{eq:HOR_general}), is not affected and the RR rate, which is obtained by using Eq. (\ref{eq:RRP_general}), and it increases until it converges to a certain value. This result is rather useful as it can be directly used to dimension the required number of RIS in order to assure UE  connectivity without over-provisioning. Over-provisioning is also not desirable from the signaling perspective, since it   can trigger signaling towards alternative paths with RISs which are not necessary. 

\par Figure \ref{subfig:RRHO2} shows the impact of eNB density $\lambda_{eNB}$ on the performance, which only affects HO rate, as expected.  HO rate increases sharply when eNB density increases for the studied interval and exceeds the RR rate at a certain point, which implies that  HO probability is increasing. When we vary the  obstacle density $\lambda_{B}$ in Figure \ref{subfig:RRHO3} both rates decrease with increasing value of $\lambda_{B}$ due to the decrease of chances to find a path through a RIS or through a handover to another eNB. In Figure \ref{subfig:RRHO4}, we the show the impact of UE speed $d_U$ on HO and RR rates, where the direction of movement is fixed ($\xi= 45^{\circ}$). For speeds lower than 5 m/s, UE is getting closer to the serving eNB, which decreases the chances of HO caused by self-blockage or  obstacles; thus, HO rate is decreasing. For speeds higher than 5 m/s, UE moves farther from the closest point to the serving eNB, increasing HO rate, as the chances of HO increase. For RR rate values, they increase with user speed as UE is getting farther from the serving RIS until it reaches RR probability close to 1 at a speed $\sim$ 4 m/s.

Figure \ref{fig:sigrate} shows the signaling rate performance when UE is moving radially with speed $d_U$, where we analyze the impact of self-blockage angle,  density of obstacles,  number of RIS-Manager servers, and number of SGW servers. In Figure \ref{subfig:sig1}, we evaluate the signaling rate using Eq. (\ref{eq:sig_rate_final}) while varying user speed $d_U$ with the self-blockage angle $\theta=0^{\circ}, 45^{\circ}, 90^{\circ}, 120^{\circ}, 180^{\circ}$. 
The signaling rate is a superposition of RR rate and HO rates shown in Figure \ref{subfig:RRHO4}. The impact of self-blockage  decreases when UE moves closer to the serving eNB and then increases when UE moves away, which can be explained by the increase of signaling related to HO affected by self-blockages. Figure \ref{subfig:sig2} shows the signaling rate as function of UE speed, with different obstacle densities from 10 obstacles per km$^2$ to 10$^5$ obstacle per km$^2$. For obstacle density smaller than 10$^3$ obstacle per km$^2$, the impact of obstacles on the signaling rate is negligible. When it increases to 10$^5$ obstacle per km$^2$, the signaling rate decreases $\sim$ 20\%, due to the lack of alternative paths for RIS reconfiguration and the lack of access to new eNB for handover. 

\par  In Figures \ref{subfig:sig3} and \ref{subfig:sig4}, we analyze the impact of number of RIS-Manager and SGW servers on the signaling rate, when the speed of UE is variable. In both figures, the signaling rate decreases with the increase of the number of servers. In figure \ref{subfig:sig3} the rate decreases by $\sim$ 10\% for a an average mobility speed of $d_U =$ 25 m/s, and by $\sim$ 6  \% for   $d_U =$ 0 m/s,  when we use 4 RIS-M servers instead of 1 server. The increase of number of RIS-M servers improves the  signaling when  the mobility speed increases, which make it interesting for dimensioning the number of server needed. For example, if a RIS-Manager server has a threshold of 55 req/s, we can satisfy this condition with 2 servers with an average speed of 10 m/s, and we need 4 servers with 15 m/s speed.s We can apply the same analysis for the dimensioning of SGW servers.

\begin{figure*}
 \centering 
 \begin{subfigure}[t]{0.24\linewidth}
 \includegraphics[scale=0.28]{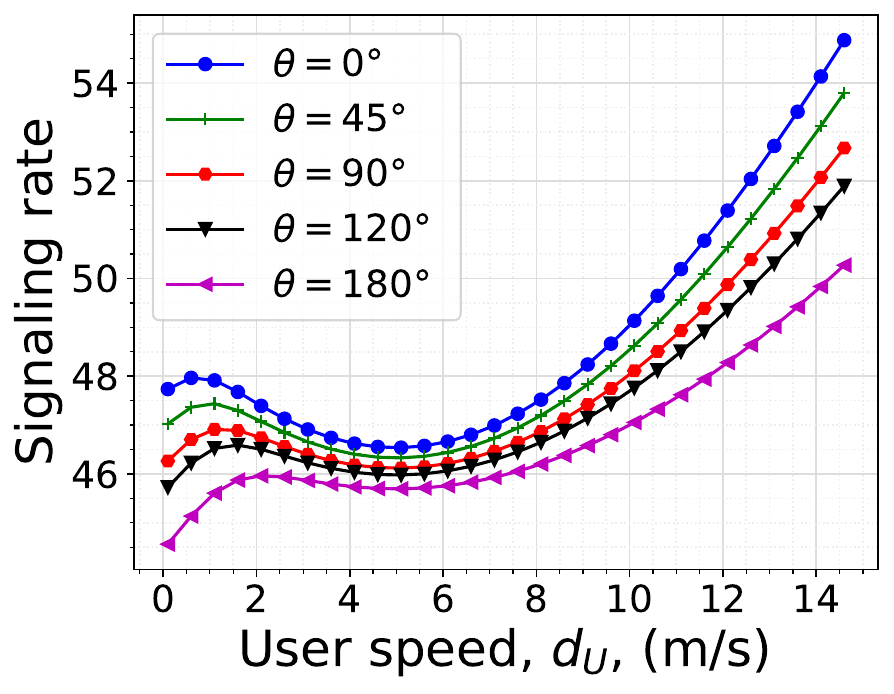}
  \caption{$\theta$  }\label{subfig:sig1}
   \end{subfigure}
   \begin{subfigure}[t]{0.24\linewidth}
 \includegraphics[scale=0.28]{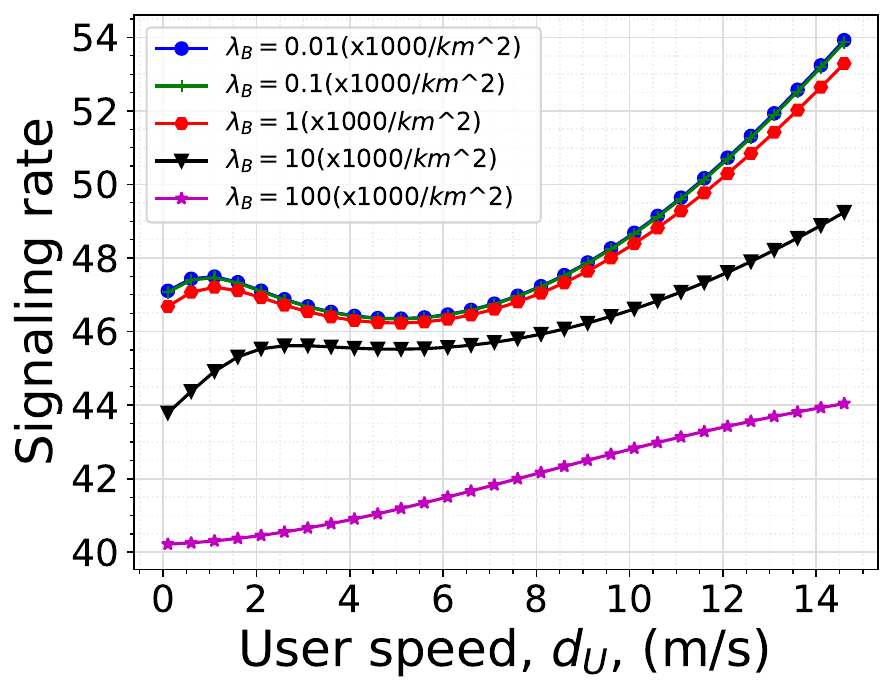}
 \caption{$\lambda_B$  }\label{subfig:sig2}
   \end{subfigure}
 \begin{subfigure}[t]{0.24\linewidth}
 \includegraphics[scale=0.28]{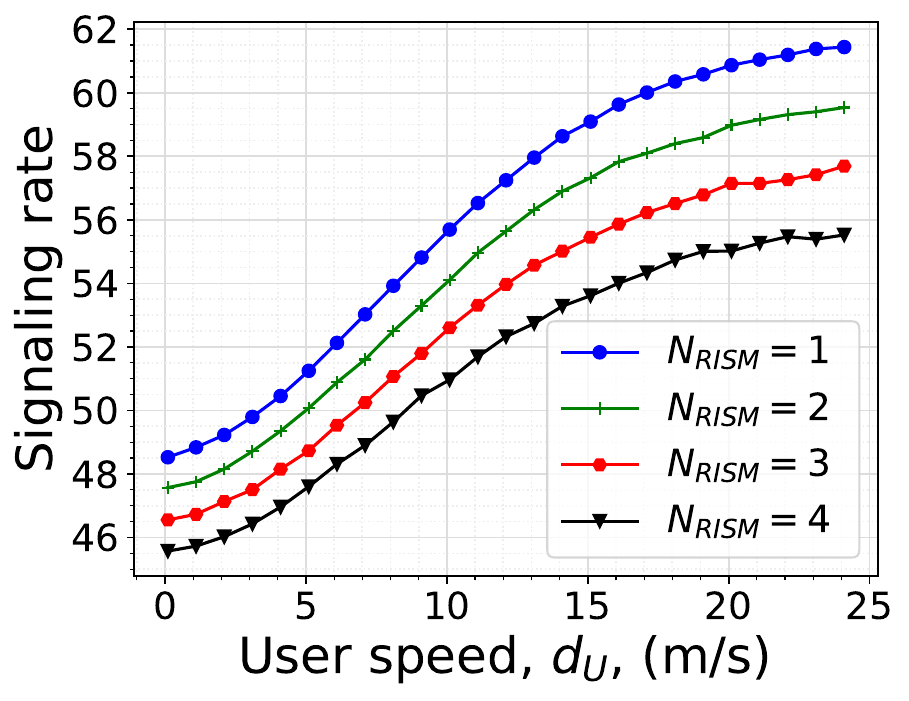}
 \caption{$N_{RISM}$  }\label{subfig:sig3}
   \end{subfigure}
 \begin{subfigure}[t]{0.24\linewidth}
  \includegraphics[scale=0.28]{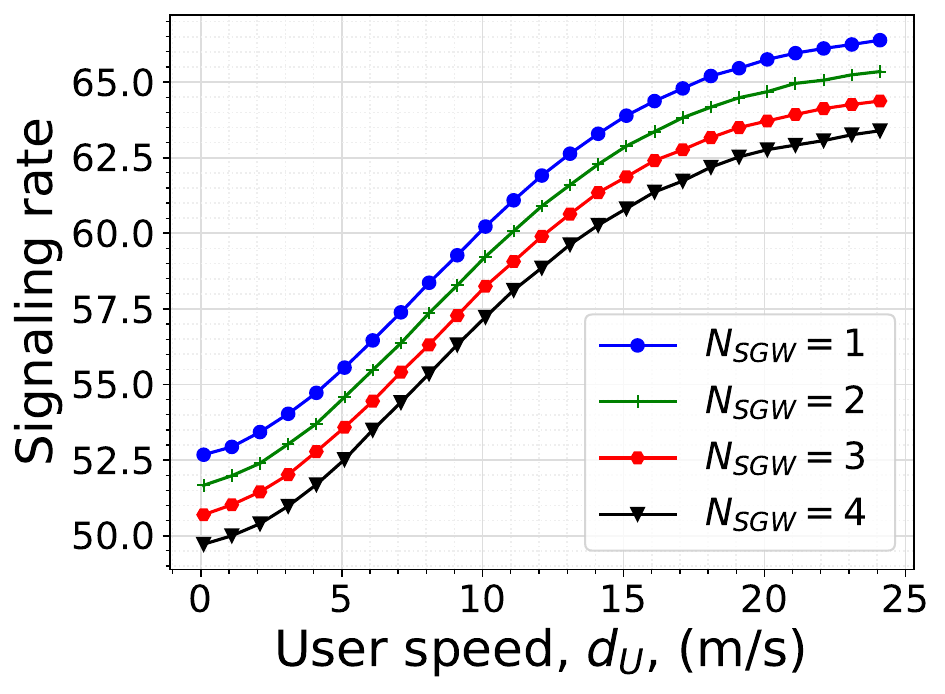}
  \caption{$N_{SGW}$  }\label{subfig:sig4}
   \end{subfigure}
 \caption{Signaling rate  in terms of static user speed.}
\label{fig:sigrate}
\vspace{-0.2 cm}
\end{figure*}

\section{Conclusion}\label{sec:conclusion}
In this paper, we provided a novel analysis of RIS reconfiguration and handover rates, defined as the dynamic reallocation of inadequate wireless paths in response to the availability of alternative paths with sufficient quality. Based on a stochastic geometry model, we analyzed these rates under various conditions and derived closed-form expressions that include known and unknown obstacles, consideration of self-blockage, RISs location density and variations in angle and speed distributions, including multiple RISs and base stations. We finally proposed and analyzed the related signaling rates of a sample signaling protocol, designed as an extension of today's known handover signaling standard and analyzed the related signaling rate theoretically.  

The results quantify the impact of obstacles on the RIS and handover reconfiguration rate as function of device density and mobility. The proposed analysis is shown as useful when dimensioning and evaluating the signaling overhead due to RIS reconfigurations, which directly impacts the planning of the related control plane server capacity in the network management system. The results show that for a certain server capacity, and an average mobility speed, our model can give the approximate needed number of RIS-manager and SGW servers. We showed examples of RIS-Manager server with maximum capacity of 55 req/s, whereby 2 servers are needed with an average user speed of 10 m/s. When we increase the number of obstacles from 10$^3$ obstacle per km$^2$ to 10$^5$ per km$^2$, the signaling rate decreases $\sim$ 20\%. Such dimensioning conclusions can be obtained from our model.

The model was shown however to have several limitations which we plan to address in future work. The model is a distance based model, without  considering the received signal strength indicator (RSSI) and SNR values to achieve handover, which will be practical for real implementations. Our signaling model is simplified to account for keep alive messages and additional signaling related to RIS reconfigurations and handover. More signaling details can be considered and analyzed such as the signaling delay performance, which is critical for real time and high speed mobile networks. Furthermore, other mobility models need to be analyzed in relations to the signaling protocols and performance.


\section*{Acknowledgment}
This work was partially supported by the DFG Project Nr. JU2757/12-1,  and by the Federal Ministry of Education and Research of Germany, joint project 6G-RIC, 16KISK031.
\vspace{-0.2cm}
\bibliographystyle{IEEEtran}
\bibliography{mybib}
\vspace{-0.2cm}
\appendix

\section{HO Rate}
\subsection{Proof of Theorem \ref{thm:HO}}\label{appedix:proof_th}

\begin{proof}
The probability that eNB-UE is not blocked by a static obstacle or by self-blockage can be obtained considering that the blockages are independent. We denote by $Z_i$ the event that eNB $i$ is located in the disk of origin $o_t$  and radius $R_{eNB}$,  and is not blocked  by static obstacles or self blockage.  Using the expression defined in the assumptions (\ref{eq:staticblockage}) and (\ref{eq:selfblockage}), the probability that eNB $i$ is not blocked is given   by:   \\    
\begin{equation}\begin{split}
P(Z_i) =&(1-P(eNB^{\text{self}}) ) \\ & \times \int_{0}^{R_{LoS}} ( 1- P(eNB_i| r) ) f_{R_{LoS}}(r)   \partial r\\
     =& (1-\frac{\theta}{2\pi})  \int_{0}^{R_{LoS}} \frac{2r}{R_{LoS}^{2}}   e^{-(\beta r + \beta_0)}  \partial r\\
     = & (1-\frac{\theta}{2\pi})  \frac{2 e^{-\beta_0}}{\beta^{2}R_{LoS}^{2}}\left[ 1- (1+ \beta R_{LoS} ) e^{-\beta R_{LoS}} \right]
\end{split}
\end{equation}                
Assuming that $n$ eNB can be located in the disk with the original position of UE and radius $R_{eNB}$, each eNB, $i$,  can be blocked with a probability $P(Z_i)$. Thus the probability that $m$ eNB are not blocked by static obstacles or self-blockage follows binomial distribution, i.e.,
\begin{equation}
P(m|n) = C_{n}^{m} P(Z)^{m}\left( 1-P(Z) \right)^{n-m}
\end{equation}
Given that the number of eNBs $N$ follows PPP distribution in an area of $\mathbb{R}^2$, the probability distribution in a disk can be given as follows and based on \cite{jain2019impact}:
\begin{equation}
P_{N}(n)= \frac{(\lambda_{eNB} \pi R_{eNB}^{2} )^{n}}{n!} e^{-\lambda_{eNB} \pi R_{eNB}^{2}}
\end{equation}
Using the distribution of $N$ and the probability that $m$ eNBs are not blocked given $n$ eNBs, the  probability  that $m$ eNBs are not blocked  can be expressed as follows:
\begin{equation}\begin{split}
P(m) &= \sum_{n=0}^{\infty} P(m|n) P_{N}(n)\\
=  &\frac{(\lambda_{eNB} \pi R_{eNB}^{2} )^{m}}{m! } P(Z)^{m} e^{-\lambda_{eNB} \pi R_{eNB}^{2}} \\
 &\times  \sum_{n=0}^{\infty}     \frac{(\lambda_{eNB} \pi R_{eNB}^{2} )^{n-m}}{(n-m)!} \left( 1-P(Z) \right)^{n-m}  \\
 =  &\frac{(\lambda_{eNB} \pi R_{eNB}^{2} )^{m}}{m! } P(Z)^{m} e^{-\lambda_{eNB} \pi R_{eNB}^{2}} \\
 & \times \sum_{n-m=0}^{\infty}     \frac{\left(\left( 1-P(Z) \right) \lambda_{eNB} \pi R_{eNB}^{2} \right)^{n-m}}{(n-m)!}  \\
= &  \frac{\left(P(Z)\lambda_{eNB} \pi R_{eNB}^{2} \right)^{m}}{m! }    e^{   -P(Z)   \lambda_{eNB} \pi R_{eNB}^{2}  } 
\end{split}
\end{equation}

The HO probability can expressed as   the probability that $m>1$ eNBs are in the disk of origin $o_t$  and radius $R_{eNB}$, given as: 

\begin{equation}\begin{split}
P(HO) &=  P(m>1) \\&= \sum_{m=1}^{\infty} P(m)\\
=& \left[ \sum_{m=0}^{\infty} \frac{\left(P(Z)\lambda_{eNB} \pi R_{eNB}^{2} \right)^{m}}{m! }   - 1\right] e^{   -P(Z)   \lambda_{eNB} \pi R_{eNB}^{2}  } \\
=& \left[ e^{P(Z)\lambda_{eNB} \pi R_{eNB}^{2}}- 1\right]e^{   -P(Z)   \lambda_{eNB} \pi R_{eNB}^{2}  } \\
=& 1 - e^{   -P(Z)   \lambda_{eNB} \pi R_{eNB}^{2}  } 
\end{split}
\end{equation}

\end{proof}

\section{RIS Reconfiguration Rate}
\subsection{Proof of Proposition \ref{prop:RRR_no_obstacle}}\label{appedix:proof1}

We define the excess area $A_1$ as shown in Figure \ref{fig:RR_staticBlocking} as:
\begin{equation}\label{eq:A1form}
A_1 = A_E - A_B
\end{equation}
where $A_E$ and $A_B$ denote, respectively, the excess area   generated by the movement of the user, and the candidate blocked area representing intersection between the circle with radius $R$ and the segment of circle of origin the position of base station and a radius $>d_{OU}+R$ and intersecting the two borders of the obstacle.
\begin{figure*}
 \centering 
   \includegraphics[scale=0.16]{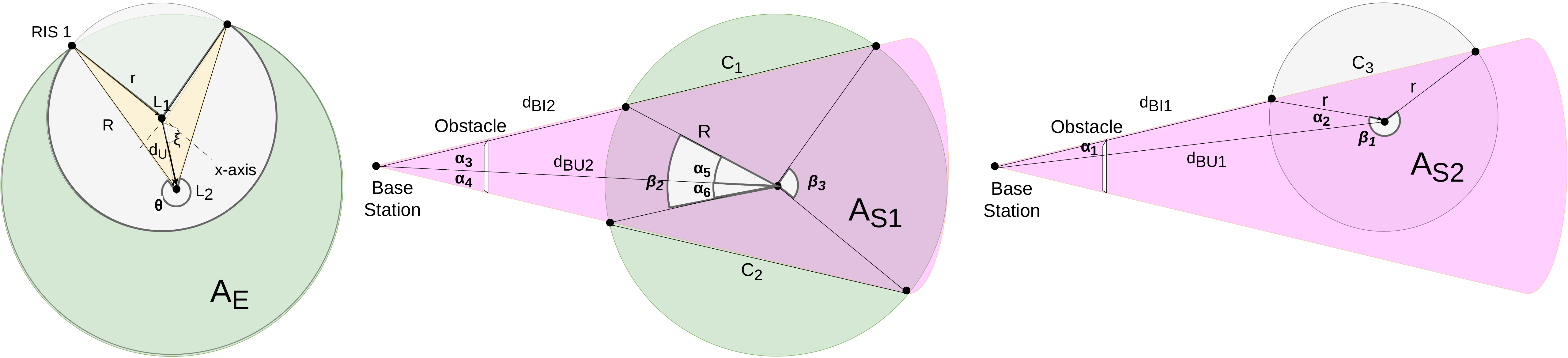}
 \caption{Steps to obtain the visible excess are for a mobile user  in RIS-aided system  with no blockage between base station and RIS or RIS and user. }
\label{fig:areas11}
\vspace{-0.5 cm}
\end{figure*}

First, in order to calculate the area  $A_E$, shown in the first illustration in  Fig. \ref{fig:areas11},  we define $A_{arc_1}$ the area generated by the arc of circle with angle $\theta_1$ of origin the location L$_2$ and radius $R$, $A_{arc_2}$ the area generated by the arc of circle with angle $\theta_2$ of origin the location L$_1$ and radius $r$, and  $A_{triangle_1}$ the area of the triangle  in  Figure \ref{fig:areas11}, with sides $R$, $r$ and $d_U$. Thus, $A_E$ can be calculated as follows:
\begin{equation}
A_E = A_{arc_1} - A_{arc_2} + 2 A_{triangle_1}
\end{equation}
The raduis $R$ of the circle of origin the location L$_2$  can be calculated in terms of the speed of the user,  the distance $d_{UR_{1}} = r$ between the operating RIS and the user at location L$_1$, and the angle of movement $\xi$. Using the  law of cosines in triangle L$_1$L$_2$R$_1$ , the expression of $R$ is given as:
\begin{equation}\label{eq:R_from_r}
R= \sqrt{d_{UR_{1}}^2 + d_{U}^2 - 2d_{UR_{1}}d_{U}\cos(\pi - \xi)}
\end{equation} 
Let $\gamma$ be the angle $\widehat{R_{1}L_2L_1}$. The angle $\widehat{L_1R_{1}L_2}$ can be expressed in terms of $\xi$ and $\gamma$ as $\widehat{L_1R_{1}L_2} = \pi -(\gamma+(\pi -\xi))=\xi -\gamma$.  Using the law of sines in the triangle  L$_1$L$_2$R$_1$ , the expression of $\gamma$ is given as: 
\begin{equation}
\gamma = \xi -sin^{-1}\left(\frac{d_{U}\sin(\xi)}{R} \right)
\end{equation} 
The angle $\theta_1$ can then be expressed as:
\begin{equation}
\begin{split}
\theta_1 =  2\pi - 2\gamma
= 2\pi - 2 \left( \xi -sin^{-1}\left(\frac{d_{U}\sin(\xi)}{R} \right) \right)
\end{split}
\end{equation}
The angle $\theta_2$ can then be expressed as:
\begin{equation}
\theta_2 = 2\pi - 2\xi
\end{equation}
Using  $\theta_1$, the area $A_{arc_1}$ can be expressed as:
\begin{equation}\begin{split}
A_{arc_1} = \frac{R^2 \theta_1}{2}
=  \pi R^2 - R^2\left( \xi -sin^{-1}\left(\frac{d_{U}\sin(\xi)}{R} \right) \right)
\end{split}
\end{equation}
The area $A_{arc_2}$ can be expressed as:
\begin{equation}\begin{split}
A_{arc_2} = \frac{r^2 \theta_2}{2}
= d_{UR_k} ^2 (\pi - \xi)
\end{split}
\end{equation}
The area $A_{triangle_1}$ can be expressed as:
\begin{equation}\begin{split}
A_{triangle_1} =& \frac{d_{UR_{k}} d_{U}\sin(\xi) }{2} 
\end{split}
\end{equation}
Thus, $A_E$ is given by:
\begin{equation}\label{eq:AE}\begin{split}
A_E = &\pi R^2 - R^2\left( \xi -sin^{-1}\left(\frac{d_{U}\sin(\xi)}{R} \right) \right)\\& -d_{UR_k} ^2 (\pi - \xi) + d_{UR_{k}} d_{U}\sin \xi 
\end{split}
\end{equation}
\begin{figure*}
 \centering 
   \includegraphics[scale=0.14]{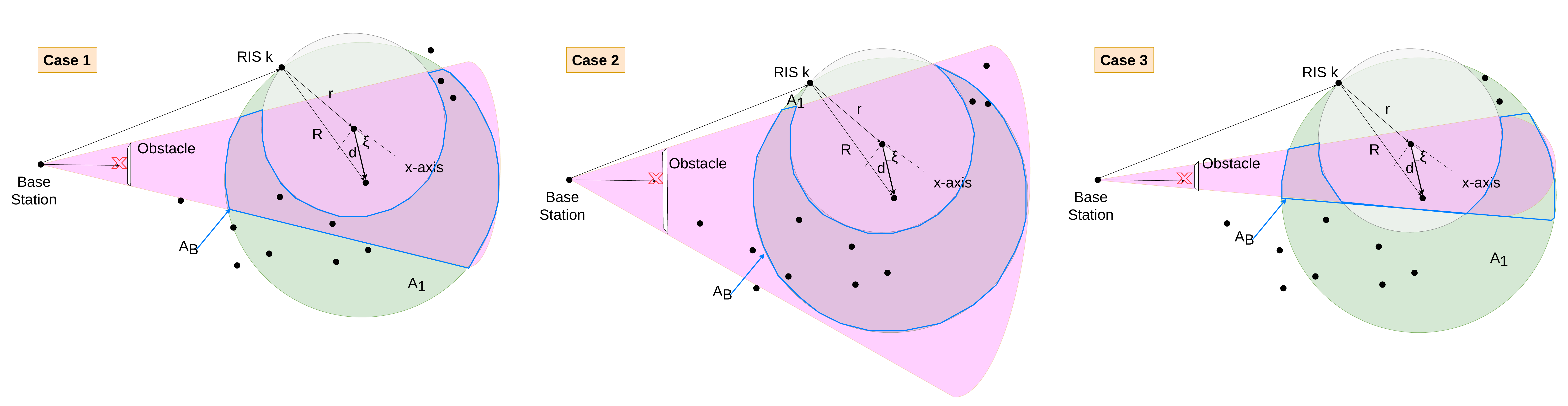}
 \caption{Blocked area situations based on obstacle position.}
\label{fig:areas1}
\vspace{-0.5 cm}
\end{figure*}
Second, in Figure \ref{fig:areas1},  we show the three possible areas, which will be explained below through calculations.  To calculate the area $A_{B}$,  we define: $A_{S_1}$ the area of intersection between the circle with radius $R$ the arc defined by the origin the base station $B$ and the borders of the obstacle as presented in Figure \ref{fig:areas11}, and $A_{S_2}$ is the area defined similar to  $A_{S_1}$ for the circle with radius $r$ presented in the third illustration of Fig. \ref{fig:areas11}. The area $A_{B}$ can be given as:
\begin{equation}
A_{B}= A_{S_1} - A_{S_2}
\end{equation}

The area $A_{S_2}$ can be composed of area $A_{\beta}$ of the  arc of angle $\beta$ and origin L$_1$, and the area $A_{triangle_2}$ of the triangle created by the origin and the two intersections $I_1$ and $I_2$. 
\begin{equation}
A_{S_2}= A_{\beta} + A_{triangle_2}
\end{equation}
We define $C_3$ the cord obtained by the intersection of the circle and the arc generated by the obstacle, $\alpha_1$ the angle between $B$, L$_1$ and the border of the obstacle. Using the law of cosines, for both triangles $B$L$_1$I$_1$ and $B$L$_1$I$_2$, we obtain the following equation and two solutions for both distance  $d_{BC1}$ and $d_{BC1} + C3$, then we conclude the value of $C_3$, as follows:
\begin{equation}\begin{split}
 X^2 &- 2d_{BU1}X\cos(\alpha_1)+ d_{BU1}^2  - r^2 = 0\\
  \Rightarrow & d_{BI1} =  d_{BU1} \cos(\alpha_1) -\sqrt{d_{BU1}^{2} \cos^{2}(\alpha_1)-d_{BU1}^2  + r^2 }, \\
  d_{BI1}& + C3 =  d_{BU1} \cos(\alpha_1) +\sqrt{d_{BU1}^{2} \cos^{2}(\alpha_1)-d_{BU1}^2  + r^2 }\\
  \Rightarrow &C3 = 2\sqrt{d_{BU1}^{2} \cos^{2}(\alpha_1)-d_{BU1}^2  + r^2 }
\end{split}
\end{equation}
Using the  Heron's formula, the area $A_{triangle_2}$ is the given by:
\begin{equation}\begin{split}
& A_{triangle_2}= (s-r)\sqrt{s(s-C3)},\\& \text{   and } s= \frac{2r+ C3}{2}= r+ \sqrt{d_{BU1}^{2} \cos^{2}\alpha_1-d_{BU1}^2  + r^2 }\\
\Rightarrow & A_{triangle_2}=  \\ &d_{BU1}  \sqrt{(d_{BU1}^{2} \cos^{2}\alpha_1 -d_{BU1}^2  + r^2 )(1-\cos^{2}\alpha_1)} 
\end{split}
\end{equation}

We define $\alpha_2$ the angle  $\widehat{R_{1}L_1I_1}$, which can be calculated using the law of cosines. The angle $\beta$ is then given as:
 \begin{equation}\begin{split}
& \alpha_2 = \cos^{-1}\left(\frac{r^2 + d_{BU1}^2 - d_{BI_1}^2}{2rd_{BU1}}\right),  \beta=\pi + 2\alpha_2\\
\Rightarrow & A_{\beta}=  \frac{\pi}{2}r^2 + r^2 \cos^{-1}\left(\frac{r^2 + d_{BU1}^2 - d_{BI_1}^2}{2rd_{BU1}}\right)
\end{split}
\end{equation}

Thus,  area $A_{S_2}$ can have two cases, either the shape of $A_{S_1}$ with given angles can be given by replacing $R$ by $r$ and $d_{BU_2}$ by $d_{BU_1}$ in the formula of  $A_{S_1}$, and will be calculated next, or given as follows:
\begin{equation}\label{eq:AS2}
\begin{split}
A_{S_2}&= \frac{\pi}{2}r^2 + r^2 \cos^{-1}\left(\frac{r^2 + d_{BU1}^2 - d_{BI_1}^2}{2rd_{BU1}}\right) \\&+ d_{BU1}  \sqrt{(d_{BU1}^{2} \cos^{2}(\alpha_1)-d_{BU1}^2  + r^2 )(1-\cos^{2}(\alpha_1))} 
\end{split}
\end{equation}

\par Area $A_{S_1}$ can exhibit two options depending on the size of the obstacle. If the obstacle hides the lower part of the circle of origin L$_2$ and radius $R$ than $A_{S_1}$ is given by Eq. (\ref{eq:AS2}), while  replacing $r$ by $R$ and $d_{BU_1}$ by $d_{BU_2}$. Otherwise,  $A_{S_1}$ can be split in four areas: (i) area of triangle $A_{triangle_3}$ defined akin to $A_{triangle_2}$ by replacing $r$ with $R$, $\alpha_1$ with $\alpha_3$ and $d_{BU_1}$ with $d_{BU_2}$, (ii) area of triangle $A_{triangle_4}$  defined similar to $A_{triangle_3}$ by changing the angle $\alpha_3$ to $\alpha_4$, (iii) area of arc $A_{\beta_2}$ of origin L$_2$ and radius $R$ and angle $\beta_2$, and (iv) an area of arc $A_{\beta_3}$ of origin L$_2$ and radius $R$ and angle $\beta_3$ as shown in Figure \ref{fig:areas11}. Thus, the area $A_{S_1}$ is given by:
\begin{equation}\label{eq:AS2_form}
A_{S_2}= A_{triangle_3} + A_{triangle_4} + A_{\beta_2} + A_{\beta_3}
\end{equation}
The area  $A_{triangle_3}$ and  $A_{triangle_4}$ are given by:
\begin{equation}\begin{split}
 &A_{triangle_3}=   \\&d_{BU2}  \sqrt{(d_{BU2}^{2} \cos^{2}(\alpha_3)-d_{BU2}^2  + R^2 )(1-\cos^{2}(\alpha_3))} \\
  &A_{triangle_4}=   \\&d_{BU2}  \sqrt{(d_{BU2}^{2} \cos^{2}(\alpha_4)-d_{BU2}^2  + R^2 )(1-\cos^{2}(\alpha_4))}
\end{split}
\end{equation}
Let $\alpha_5$ and $\alpha_6$ be the angles between $B$, L$_2$ and the top left intersection and the bottom left intersection, both angles can be calculated similar to $\alpha_2$, using $R$ instead of $r$, $d_{BU_2}$ instead of $d_{BU_1}$, and $d_{BI_3}$ and $d_{BI_4}$ instead of $d_{BI_1}$. Area $ A_{\beta_2} $ is given as follows:
\begin{equation}\label{eq:Abeta2_form}\begin{split}
 d_{BI_3} = & d_{BU_2} \cos \alpha_3 -\sqrt{d_{BU_2}^{2} \cos^{2} \alpha_3 -d_{BU_2}^2  + R^2 }, \\ d_{BI4} = & d_{BU_2} \cos \alpha_4  -\sqrt{d_{BU_2}^{2} \cos^{2} \alpha_4 -d_{BU_2}^2  + R^2 }  \\ 
\alpha_5 = &\cos^{-1}\left(\frac{R^2 + d_{BU_2}^2 - d_{BI_3}^2}{2Rd_{BU_2}}\right), \\ \alpha_6 = & \cos^{-1}\left(\frac{R^2 + d_{BU_2}^2 - d_{BI_4}^2}{2Rd_{BU_2}}\right)\\
\Rightarrow A_{\beta_2}= &\frac{R^2 (\alpha_5 + \alpha_6)}{2} \\ =&\frac{R^2 }{2} \left( \cos^{-1}\left(\frac{R^2 + d_{BU_2}^2 - d_{BI_3}^2}{2Rd_{BU_2}}\right) \right.\\ & \left. + \cos^{-1}\left(\frac{R^2 + d_{BU_2}^2 - d_{BI_4}^2}{2Rd_{BU_2}}\right) \right)\\
\end{split}
\end{equation}
Let $\beta_4$ and $\beta_5$ be the  angles of origin L$_2$ and the two point of top and bottom intersections between the blocked area and the circle with center  L$_2$ and radius $R$. Both angles can be concluded in terms of $\alpha_3$,  $\alpha_4$, $\alpha_5$, and $\alpha_6$. Thus, we can obtain $\beta_3$, and the area $A_{\beta_3}$ as follows:
\begin{equation}\label{eq:Abeta3_form}\begin{split}
 \beta_4 = & \pi - 2(\alpha_3 + \alpha_5 ),  \beta_5 =  \pi - 2(\alpha_4 + \alpha_6 )  \\ 
\Rightarrow  \beta_3 = & 2\pi - \beta_2-\beta_4-\beta_5=2\alpha_4 + 2 \alpha_3 + \alpha_5 + \alpha_6\\
\Rightarrow A_{\beta_3}=& \frac{R^2 \beta_3}{2} \\  =&\frac{R^2 }{2} \left[ 2\alpha_4 + 2 \alpha_3 +  \cos^{-1}\left(\frac{R^2 + d_{BU_2}^2 - d_{BI_3}^2}{2Rd_{BU_2}}\right) \right. \\ & \left. + \cos^{-1}\left(\frac{R^2 + d_{BU_2}^2 - d_{BI_4}^2}{2Rd_{BU_2}}\right) \right]\\
\end{split}
\end{equation}
Therefore, area $A_{S_1}$ is given as:
\begin{equation}\label{eq:AS1}\begin{split}
A_{S_1}=& A_{triangle_3} + A_{triangle_4} + A_{\beta_2} + A_{\beta_3} \\ 
= &   d_{BU_2} \left[ \sqrt{(d_{BU_2}^{2} \cos^{2}\alpha_3-d_{BU_2}^2  + R^2 )(1-\cos^{2}\alpha_3)} \right.\\ &\left.+
   \sqrt{(d_{BU_2}^{2} \cos^{2}\alpha_4-d_{BU_2}^2  + R^2 )(1-\cos^{2}\alpha_4)} \right]
\\ &+ R^2  \left[ \alpha_4 +  \alpha_3 +  \cos^{-1}\left(\frac{R^2 + d_{BU_2}^2 - d_{BI_3}^2}{2Rd_{BU_2}}     \right) \right. \\ & \left. + \cos^{-1}\left(\frac{R^2 + d_{BU_2}^2 - d_{BI_4}^2}{2Rd_{BU_2}}\right)  \right]
\end{split}
\end{equation}
The closed form of area $A_{B}$ can be expressed in terms of distance between the user $U$  and the base station $B$ and the angles of the blocked area due to the obstacle, as follows:
\begin{equation}\begin{split}\label{eq:AB}
A_{B}= &A_{S_1} - A_{S_2}\\=&
 d_{BU_2} \left[ \sqrt{(d_{BU_2}^{2} \cos^{2}\alpha_3-d_{BU_2}^2  + R^2 )(1-\cos^{2}\alpha_3)} \right. \\ & \left. +
   \sqrt{(d_{BU_2}^{2} \cos^{2}\alpha_4-d_{BU_2}^2  + R^2 )(1-\cos^{2}\alpha_4)} \right]\\&  - d_{BU1}  \sqrt{(d_{BU1}^{2} \cos^{2}\alpha_1 -d_{BU1}^2  + r^2 )(1-\cos^{2}\alpha_1)} 
\\ &+ R^2  \left[ \alpha_4 +  \alpha_3 +  \cos^{-1}\left(\frac{R^2 + d_{BU_2}^2 - d_{BI_3}^2}{2Rd_{BU_2}}     \right) \right.\\ & \left.+ \cos^{-1}\left(\frac{R^2 + d_{BU_2}^2 - d_{BI_4}^2}{2Rd_{BU_2}}\right)  \right]- \frac{\pi}{2}r^2  \\
& - r^2 \cos^{-1}\left(\frac{r^2 + d_{BU1}^2 - d_{BI_1}^2}{2rd_{BU1}}\right) 
\end{split}
\end{equation}
Finally, the closed form of the excess area $A_1$ is concluded by substituting Eqs. (\ref{eq:AE}) and (\ref{eq:AB})  into Eq. (\ref{eq:A1form}).

\end{document}